%% file: main.tex
\documentclass[12pt]{article}

\usepackage[slantedGreek]{mathpazo}
\usepackage{microtype}
\usepackage{graphicx}
\usepackage{amsfonts}
\usepackage{setspace}
\usepackage{amssymb}
\usepackage{graphicx}
\usepackage{amsmath}
\usepackage{amsthm}
\usepackage{mathtools}
\usepackage{lscape}
\usepackage{subcaption}
\usepackage{suffix}
\usepackage{color}
\usepackage{bm}
\usepackage{verbatim}
\usepackage{dsfont}

\usepackage[authoryear, sort, comma]{natbib}

\definecolor{darkblue}{rgb}{0.0,0.0,0.55}
\RequirePackage[colorlinks,citecolor=darkblue,urlcolor=darkblue,linkcolor=darkblue]{hyperref} 
\usepackage[normalem]{ulem}

\newtheorem{theorem}{\textsc{{Theorem}}}

\newtheorem{lemma}{\textsc{{Lemma}}}

\newtheorem{proposition}{\textsc{{proposition}}}
\newtheorem{definition}{\textsc{{Definition}}}
\theoremstyle{remark}
\newtheorem{remark}{{\upshape\textsc{Remark}}}
\newtheorem{assumption}{{\upshape\textsc{Assumption}}}

\usepackage{booktabs}
\usepackage{array}
\usepackage{multirow}
\usepackage{enumitem}
\input{commands}

\setcitestyle{citesep={;}}

\allowdisplaybreaks

\usepackage{algorithm, float}
\usepackage{algpseudocode}
\newlength{\continueindent}
\title{\vspace{-1cm} \baselineskip=20pt  \bf Bayesian Soft-Tree Poisson Process Model for  Covariate-Dependent Point Patterns}

\author{Pulong Ma \\
Department of Statistics, Iowa State University \\
Ames, IA 50011, USA}
\date{}
\begin{document}
\maketitle
\baselineskip=15pt

\begin{abstract}
Estimating covariate-dependent intensity functions from spatial point pattern is challenging when the integrated intensity function cannot be evaluated analytically. This challenge becomes even worse in Bayesian inference due to repeated integration of the intensity function over the high-dimensional covariate space. This paper develops a Bayesian soft-tree Poisson process model that combines adaptive covariate partitioning with smooth intensity function. The key idea is to use soft gating functions to smooth transitions between partition boundaries in a tree generating process while preserving analytic evaluation of integrated intensity. A Bayesian CART prior is adopted to learn the partition structure from data. Given the partition tree, conjugate priors are assigned on leaf parameters to enable analytic marginal likelihood computation. This paper also develops  a reversible-jump Markov chain Monte Carlo (RJ-MCMC) algorithm for posterior inference. Posterior consistency is established for the proposed model with respect to the Hellinger distance. Extensive simulation studies and an application to wildfire data confirm the advantages of the proposed method over existing methods. 
\end{abstract}

\section{Introduction}

Spatial point patterns are often collected in many scientific disciplines including climate science, ecology, biology, seismology \citep{Cressie1993, Diggle2014, Gelfand2018}. Statistical analysis of these point pattens often relies on point processes. A popular class of point processes is Poisson (point) process (PP). Let $X\sim \mathcal{PP}(\lambda)$ denote a Poisson process with intensity function $\lambda$ on a domain $D\subset \mathbb{R}^d$, where $\lambda: D \mapsto (0, \infty)$ is a locally integrable function.  A constructive definition for  PP is as follows.  For any bounded subset $B\subset D$, the number of points in $B$, denoted by $N(B)$, is Poisson distributed with rate $\Lambda(B):=\int_B \lambda(\bfx) \mathrm{d} \bfx$; given $N(B)$, the points in $B$ are independently and identically distributed with \textit{location density} $f(\bfx):=\lambda(\bfx) / \Lambda(B)$. Here $\Lambda(B)$ is referred to as total intensity (over $B$); see \citet{Moller2004} for some backgrounds. If $\lambda(\cdot)$ is constant over $D$, $X$ is called a homogeneous PP, otherwise $X$ is an inhomogeneous PP. 

Modeling spatial point patterns requires flexible modeling for the intensity function. Given a realization $\bfx_{1:n}:=\{\bfx_1, \ldots, \bfx_n\}\subset D$ of $X$, the likelihood given the intensity function is $L(\lambda; \bfx_{1:n}) = \exp(-\Lambda(D)) \prod_{i=1}^n \lambda(\bfx_i)$. Several statistical methods have been proposed to model the intensity function. A well-known approach is to model  the intensity function with a latent Gaussian process via a link function, resulting in so called  the log-Gaussian Cox process (LGCP) \citep{Moller1998} when a log link function is used. In LGCP, the mean function of the Gaussian process (GP) often incorporates covariates and the covariance function is used to capture spatial dependence. This class of models is often applied to environmental data with spatial locations, where covariance function is defined over  a spatial domain.  A potential limitation of this approach to deal with multiple (or high-dimensional) covariates is that these covariates are required to be known to calculate the random integral. 

To perform inference, a standard practice to compute this random integral is to use quadrature approaches with a dense grid points over spatial domain. Most often, the covariates are sparsely observed over finite number of spatial locations. A preprocessing step is often used to interpolate these covariates over the spatial domain so that the numerical integral can be approximated over the quadrature grids; see \citet{Moller1998, Rue2009, Simpson2016}. Potential challenges in LGCP include the numerical approximation error for approximating the random integral and the lack of uncertainty quantification when interpolated environmental variables are used. For instance, GP related challenges arise when the number of quadrature grid points is large and the underlying smoothness of intensity is controlled by the smoothness of GP, which is often chosen to be second-order stationary in practice despite the fact that environmental processes often have heterogeneous spatial dependence structure over the entire domain.  

To address the computational challenges associated with LGCP, several strategies have been proposed in the literature. For GP challenges, efficient GP models such as Gaussian Markov random field \citep{Rue2009} have been used  to exploit sparse matrix operations for efficient computations \citep{Rue2009, Simpson2016}. To bypass the numerical integration problem in likelihood, data augmentation techniques are also exploited to allow efficient inference \citep{Donner2018, Adams2009}, where sigmoid GP is used for modeling intensity function via the sigmoid link function.

Another approach for modeling the intensity functions is focused on domain partitioning methods. Earlier methods include Voronoi tessellation  \citep[e.g.,][]{Meijering1953, Hinde1980, Barr2010}. More recently, tree-based methods \citep{Ma2026PPT, Lu2026, Lamprinakou2023} have been used for modeling intensity functions.   Voronois tessellation method often assumes the piecewise constant over Voronoi cells; in contrast, tree-based methods represent the intensity function over a single tree or additive trees that are formed by random recursive partitioning of the domain so that partition tree can be learnt from the data. 

In particular, \citet{Ma2026PPT} represents the intensity function as piecewise constants in a single tree structure that is formed by Bayesian CART prior \citep{Chipman1998, Chipman2002}. \cite{Lamprinakou2023} develop a statistical model called BART-PP that uses  Bayesian additive regression trees \citep[BART;][]{Chipman2010} for log-intensity functions. A non-Bayesian formulation is given by \cite{Lu2026} who develop a statistical model called XGBoostPP that uses XGBoost \citep{Chen2016, Natekin2013}. {Additive tree ensembles such as XGBoostPP and BART-PP are generally more flexible than single tree models for function approximation.} A potential limitation for XGBoostPP is that it does not provide uncertainty estimates for intensity function from the model; in contrast, BART-PP can provide Bayesian uncertainty quantification. However, BART-PP does not learn hyperparameters in leaf nodes but fixes them {\it a priori}, which would underestimate the uncertainty. A key benefit of tree-based methods is its ability to deal with potentially high-dimensional covariate space without requiring fully observed covariates over the domain since tree-based methods such as \citet{Ma2026PPT} and \citet{Lamprinakou2023} can obtain closed form formula in evaluating the random integral in the likelihood function due to piecewise constants assumptions.

Recognizing  the benefits of tree-based models, this paper focuses on modeling the intensity function directly through a tree structure. As noted in \citet{Ma2026PPT}, a potential limitation in using a hard-split tree is that the resulting posterior of the intensity function can have very sharp transitions across partition boundaries. To address this limitation, this paper exploits the ideas \citep{Irsoy2012, Jordan1994, Linero2018} to use a soft-gating function to enable smooth transitions across partition boundaries, where the degrees of smoothness across the boundary transitions will be estimated in a Bayesian paradigm. While it may seem that this would destroy the analytic expressions of the random integral, this paper will design a strategy to maintain analytic expression for the random integral, thus reducing the integral to one-dimensional closed-form factors, with numerical fallback when needed for stability. 

To allow adaptive estimation of the partitions, this paper couples the Bayesian CART prior \citep{Chipman1998, Chipman2002} with soft gating strategy to learn the partition. This leads to a hierarchical specification for the proposed model. To make Bayesian inference, this paper develops a novel data-augmentation strategy to address the computational challenges since with soft-gating functions, the likelihood of the intensity parameter with each leaf node is not conjugate to its prior distribution. This computational challenge is addressed by introducing a latent variable for each observation to mark their memberships within a node. Under this augmentation, the augmented likelihood is shown to allow conjugate update for intensity parameters so that  Bayesian inference could be performed  efficiently.

\textbf{Contributions:} The contributions of this paper ranging from statistical method, computational algorithm, and theory. In particular, this paper builds a Bayesian soft Poisson point process tree (\textit{S-PPT}) model for modeling point patterns in a potentially high-dimensional covariate space, for which many existing methods fail to cope with in many real-world applications. The paper also develops an efficient Bayesian inference procedure through a data-augmentation technique to allow conjugate updates for intensity parameters. Finally the paper also establishes strong consistency of the resulting posterior distribution under Hellinger distance in an asymptotic framework.

\section{Statistical Method}

\subsection{Bayesian Soft Poisson Process Tree}
Let $D = \prod_{j=1}^d [a_j, b_j]$ be the domain of interest. Suppose that we observe a realization $\bfx_{1:n}:=\{\bfx_1, \ldots, \bfx_n\}$ of $X$ on $D$.  
Let $T$ denote a tree drawn from the Bayesian CART prior $\pi(T)$ with some hyperparameters. Its detailed specification is given later on. Then the intensity function $\lambda(\cdot)$ is assumed to be a piecewise smooth function over the leaf nodes of the tree. Write $\mathcal L=\mathcal L(T)$ for the leaves and $\mathcal I=\mathcal I(T)$ for the internal nodes. The proposed soft Poisson process tree (\textit{S-PPT}) model can be specified hierarchically as follows. 
\begin{gather} \label{eqn: model for X}
X \sim \mathcal{PP}(\lambda), \\ \label{eqn: model for lambda}
\lambda(\bfx) = \sum_{v \in \calL} \lambda_v \phi_v(\bfx; \gamma), \\
\lambda_v | T, \kappa, \xi \overset{iid}{\sim} \mathrm{Ga}(\kappa, \kappa/\xi), \, \gamma \sim \pi(\gamma) \\
 T \sim \pi(T), 
\end{gather}
where $\phi_v(\cdot; \gamma)$ is a soft gating function with $\gamma=(\gamma_1,\ldots, \gamma_d)^\top$ denoting the vector of all dimension-specific gating parameters with $\gamma_j\overset{iid}{\sim} \mathrm{Ga}(a_0, b_0)$.  Equation~\eqref{eqn: model for lambda} has piecewise constants  $\lambda_v$ at each leaf that are smoothed by the gating function $\phi_v(\cdot; \gamma)$.

\paragraph{Bayesian CART prior} The Bayesian CART prior \citep{Chipman1998, Chipman2002} recursively splits the domain $v=D$ according to a partition rule defined by a splitting random variable $S(v) \sim \mathrm{Bern}(\rho)$, axis selection random variable $J(v) \sim \mathrm{Cat}(\zeta)$, and a cutting point random variable $L(v)=a_j + U_v(b_j-a_j)$  with $U_v\sim F_L$, where  $\rho$ is the (prior) splitting probability, $\zeta = (\zeta_1, \ldots, \zeta_d)^\top$ is the vector of prior selection probabilities satisfying $\zeta_i>0, \sum_{i=1}^d \zeta_i=1$. In Bayesian CART, $F_L$ is chosen to be the quantiles of the covariate values along selected axis in \citet{Chipman2002}. Starting with root node $v=D$, Bayesian CART prior firsts draw $S(v) \sim \mathrm{Bern}(\rho)$ with $\rho(v):=\alpha_0 (1 + \mathrm{depth}(v))^{-\beta_0}$. If $S(v)=0$, then the node $v$ is terminal; otherwise, draw $J(v)=j$ from its prior distribution $\mathrm{Cat}(\zeta)$ to select an axis and a cutting point $U=u\sim F_L$ with $L(v)=c=a_j + u(b_j-a_j)$ so that $v$ is split into two child nodes $v_\lchild, v_\rchild$:
\[
v_{\lchild} = [a_1,b_1]\times \cdots \times [a_j, c] \times \cdots [a_d, b_d]\qquad v_{\rchild} = [a_1,b_1]\times \cdots \times [c, b_j] \times \cdots [a_d, b_d].
\]
Then this procedure is applied recursively for each child node. 

\paragraph{Soft gating} The Bayesian CART prior recursively partitions the space with hard splits at the boundaries of subregions  causing sharp changes that may not be desirable. To allow gradual transitions across split boundaries, we replace binary routing decision with a smooth gating function. Each observation thus receives weights from multiple terminal nodes rather than belonging exclusively to one terminal region. This construction comes from hierarchical mixtures of experts \citep{Jordan1994} and soft decision trees \citep{Irsoy2012, Linero2018}. 

The gating function $\phi_v(\cdot;\gamma)$ is constructed recursively.
Let $W_j=b_j-a_j>0$ be the width of the root domain along coordinate $j$.
At an internal node $u$ with splitting coordinate $j_u$ and cut $c_u$,
define the right-routing function
\[
g_u(t;\gamma)=\left\{1+\exp\left[-\gamma_{j_u}
\frac{t-c_u}{W_{j_u}}\right]\right\}^{-1},\qquad \gamma_{j_u}>0.
\]
The root width is used at every node, so all gates on coordinate $j$
share the slope $\gamma_j/W_j$. Each cut is strictly inside its node's
interval on the selected coordinate. The left-routing function is
$1-g_u(t;\gamma)$, and hard routing is recovered away from the cut as
$\gamma_{j_u}\to\infty$. Starting with $\phi_{\mathrm{root}}=1$, set
\[
\phi_{u_\lchild}(\bfx)=\phi_u(\bfx)\{1-g_u(x_{j_u};\gamma)\},\qquad \phi_{u_\rchild}(\bfx)=\phi_u(\bfx)g_u(x_{j_u};\gamma).
\]
Thus $\phi_{u_\lchild}+\phi_{u_\rchild}=\phi_u$ and
$\sum_{v\in\mathcal L}\phi_v(\bfx)=1$.
For a node $v$, let $\downarrow\!v$ contain its internal ancestors and
let $\delta_{uv}\in\{0,1\}$ indicate whether its path turns right at $u$.
Then
\[
\phi_v(\bfx)=\prod_{u\in\downarrow\!v}
g_u(x_{j_u};\gamma)^{\delta_{uv}}
\{1-g_u(x_{j_u};\gamma)\}^{1-\delta_{uv}}.
\]
The following result gives the closed-form expression for  
$I_v:=\int_D\phi_v(\bfx;\gamma)\,\mathrm d\bfx$.

\begin{proposition}\label{prop: analytic exposure}
Let $T$ be finite and use the root-width gates above on
$D=\prod_{j=1}^d[a_j,b_j]$. Then
$
I_v=\prod_{j=1}^d J_{vj}.
$
If the path has no split on coordinate $j$, then $J_{vj}=W_j$.
Otherwise, suppress $v,j$ and denote its $m$ cuts and turn indicators
by $c_1,\ldots,c_m$ and $\delta_1,\ldots,\delta_m$. Put $W=W_j$,
$\gamma=\gamma_j$, $r=\sum_{k=1}^m \delta_k$, and
\[
q_k=\exp\!\left\{\frac{\gamma}{W}
       \left(c_k-\frac{a_j+b_j}{2}\right)\right\},\qquad C=\prod_{k:\delta_k=0}q_k,\qquad
A_k=\frac{C(-q_k)^{r-1}}
          {\prod_{\ell\ne k}(q_\ell-q_k)}.
\]
The interior-cut condition makes these cuts distinct along the path,
and
\[
J_{vj}=\frac W\gamma\left\{
\gamma\mathds{1}_{\{r=0\}}+
\sum_{k=1}^m A_k\log
\frac{e^{\gamma/2}+q_k}{e^{-\gamma/2}+q_k}
\right\}.
\]
\end{proposition}

The proof is given in Appendix~\ref{app: analytic exposure proof}.
The formula applies to rectangular regions under axis-aligned partitions.  The routing identities also give $I_{u_\lchild}+I_{u_\rchild}=I_u$ and
$\sum_{v\in\mathcal L}I_v=\prod_j W_j$.

\subsection{Likelihood Computation} \label{sec: likelihood}
Given $n$ points $\bfx_{1:n}$ and model parameters $T, \bflambda, \gamma$, the likelihood function under S-PPT is given by 
\begin{equation}\label{eqn: likelihood}
p(\bfx_{1:n}|\bflambda, \gamma, T) = \exp\left(-\sum_{v\in \calL} \lambda_v I_v  \right) \times \prod_{i=1}^n \left\{ \sum_{v\in \calL} \lambda_v \phi_v(\bfx_i)  \right\},
\end{equation}
where $\bflambda:=(\lambda_v)_{v\in\calL}$ is the collection of all intensity parameters at the leaf nodes of the tree. In this likelihood function, $N(D)=n$ is omitted for notational simplicity whenever there is no scope of confusion since $N(D)=n$ is a realization from a Poisson distribution resulting from a Poisson process. 
This likelihood has a product-summation form and it can be evaluated analytically; however, the conjugate prior for $\lambda_v$ does not ease the posterior update for $\lambda_v$ since the conditional posterior $p(\lambda_v| \bfx_{1:n}, T, \gamma, \bflambda_{-v}) \propto p(\bfx_{1:n}|\bflambda, \gamma, T) \pi(\lambda_v)$ is not a standard distribution to be sampled from, where $\bflambda_{-v}$ denotes the collection of all $\lambda_u$ for $u \in \mathcal{L}\setminus \{v\}$. As $\lambda_v$ is defined for each node for a given tree, posterior inference for $\bflambda$ becomes extremely challenging and slow.

\subsection{Posterior Inference} \label{sec: posterior}
Under the likelihood function~\eqref{eqn: likelihood}, the joint posterior of $T, \gamma, \bflambda$ is 
\begin{align}\label{eqn: joint posterior}
p(T, \gamma, \bflambda | \bfx_{1:n}) \propto p(\bfx_{1:n}|\bflambda, \gamma, T)  \pi(\bflambda | T) \pi(T) \pi(\gamma). 
\end{align} 
where $\pi(\bflambda | T) = \prod_{v \in \calL} p(\lambda_v | T)$ with $\lambda_v\sim \mathrm{Ga}(\kappa, \kappa/\xi)$. 
As discussed in Section~\ref{sec: likelihood}, directly updating $\bflambda$ is challenging. 
This paper develops a data-augmentation technique \citep{Tanner1987} to address this challenge.

\subsubsection{Data Augmentation}

For each point $\bfx_i$, let $z_i$ denote the latent leaf label assigned to $\bfx_i$. Let $n_v:=\#\{i: z_i = v\}$ be the number of points in node $v$. Conditional on $N(D)=n$ and $(T,\bflambda,\gamma)$, before conditioning on the locations, the labels $z_i$ are independent with probabilities 
$
\Pr(z_i=v| T, \bflambda, \gamma) = {\lambda_vI_v }/\{\sum_{u\in\calL} \lambda_u I_u\}.
$

Given these latent variables, the likelihood function~\eqref{eqn: likelihood} can be augmented: Given all points in $v$ and the label $z_i=v$, the conditional probability of all points in $v$ has the Poisson process likelihood. Thus, given tree $T$, the augmented likelihood function over all $n$ points becomes
\begin{align} \label{eqn: augmented likelihood}
p(\bfx_{1:n}, z |\bflambda, \gamma, T) = \prod_{v\in \calL} e^{-\lambda_v I_v} \lambda^{n_v}_v \prod_{i: z_i=v}  \phi_{z_i}(\bfx_i). 
\end{align}
where $z:=(z_1,\ldots, z_n)^\top$. This augmented likelihood function has a product form over each node. It is immediate that integrating out the intensity $\lambda_v$ with respect to its prior $\pi(\lambda_v) = \mathrm{Ga}(\kappa, \kappa/\xi)$ yields a closed form expression:
\begin{equation}\label{eqn: integrated likelihood}
p(\bfx_{1:n}, z |\gamma, T) = \int p(\bfx_{1:n}, z |\bflambda, \gamma, T) \prod_{v\in \calL} \pi(\lambda_v) \mathrm{d}  \bflambda = \prod_{v \in \calL} h(n_v, I_v) \prod_{i=1}^n \phi_{z_i}(\bfx_i)
\end{equation}
where $h(m,I):= {(\kappa/\xi)^{\kappa} \Gamma(\kappa+m)}/\{\Gamma(\kappa) (\kappa/\xi + I)^{\kappa + m}\}$. This integrated likelihood function can be used to design a collapsed Gibbs sampler. The following proposition gives the collapsed full conditional distribution of each allocation $z_i$:
\begin{proposition}\label{prop: color probability} 
For $i=1,\ldots, n$, denote the number of points in $v$ except the $i$th point by $n_{v}^{-i}:=\#\{j \neq i: z_j = v\}$, the conditional probability of $z_i=v$ given $z_{-i},\bfx_{1:n},T,\gamma$ is 
$
\Pr(z_i=v| z_{-i},\bfx_{1:n},T,\gamma) \propto \phi_v(\bfx_i) ( \kappa + n_v^{-i}) / (\kappa/\xi + I_v). 
$
\end{proposition}

The proof is given in Appendix~\ref{app: color probability proof}.
The factor $( \kappa + n_v^{-i}) / (\kappa/\xi + I_v)$ is the mean of the $\mathrm{Ga}(\kappa + n_v^{-i}, \kappa/\xi + I_v)$ distribution which is the posterior of $\lambda_v$ computed from other observations. The term $\kappa + n_v^{-i}$ implies that leaves that already hold many points are more likely to receive point $i$. 

\subsubsection{Partially Collapsed Gibbs Sampler}
The joint posterior distribution of $T, z, \gamma, \bflambda$ is 
\begin{equation}\label{eqn: augmented posterior}
p(T, z, \gamma, \bflambda | \bfx_{1:n})  \propto p(\bfx_{1:n}, z |\bflambda, \gamma, T)  \pi(\bflambda | T) \times  \pi(T) \pi(\gamma)
\end{equation}
Following \citet{vanDyk2015} and \citet{vanDyk2008}, use the following partially collapsed scan. The $z$ step restores the labels after they are marginalized in the $\gamma$ step. Proposition~\ref{prop: color probability} gives a separate collapsed allocation conditional.

\begin{itemize}
\item Update $\bflambda$: Sample $\lambda_v | T, z, \gamma, \bfx_{1:n} \sim \mathrm{Ga}(\kappa + n_v, \kappa/\xi + I_v)$ for $v\in \calL$;
\item Update $\gamma$: Use a joint Metropolis--Hastings step on $\log\gamma$ targeting
\[
p(\gamma\mid T,\bflambda,\bfx_{1:n})\propto\pi(\gamma)p(\bfx_{1:n}\mid\bflambda,\gamma,T),
\]
with adaptation during warm-up only; see Appendix~\ref{app: gate update}.

\item Update $z$: Sample $z_i$ from a Categorical distribution with $\Pr(z_i=v| \bflambda, T, \gamma, \bfx_{1:n}) \propto \lambda_v \phi_v(\bfx_i) $ for $i=1,\ldots, n$. 
\item Update $(T,z)$: Holding $\gamma$ fixed and integrating out
$\bflambda$, use a joint Metropolis--Hastings update targeting
$\pi_c(T,z)\propto\pi(T)p(\bfx_{1:n},z\mid\gamma,T)$,
with the integrated likelihood in Eq.~\eqref{eqn: integrated likelihood}.
The GROW, PRUNE, and CHANGE proposals and their acceptance
probabilities are given in Appendix~\ref{app: PCGS}.

\end{itemize}
After the tree step, refresh $\bflambda$ from its Gamma full conditional
before recording a retained draw. This restores the joint posterior
from the rate-marginal state.

\section{Posterior Consistency} \label{sec: theory}
This section studies the frequentist property of Bayes posterior under the proposed model. 

\subsection{Notation and Preliminaries}
Let $\mu$ be a fixed diffuse Borel measure \citep[Proposition~9.1.III(i)]{Daley2008} on $D=[0,1]^d$ with  fixed $d$ and $0<\mu(D)<\infty$. Then  
  $\mu(\{\bfx\})=0$ for every $\bfx\in D$. Denote the space of nonnegative and integrable intensities by 
\begin{align*}
\Theta:=\left\{\lambda: D\mapsto [0, \infty): \int_D \lambda \mathrm{d}\mu <\infty \right\}.
\end{align*}
and its Borel $\sigma$-algebra under the $L^1(\mu)$ norm by $\mathcal{B}(\Theta)$. Write $\pi$ for the joint prior on $(T, \bflambda, \gamma)$ in the Soft PPT model, and $\Pi$ for the induced law of $\lambda$ on $\Theta$.

For $\lambda\in\Theta$, write
\[
\mathbb{P}_\lambda=\mathcal{PP}(\lambda\mu),\qquad
Q=\mathcal{PP}(\mu),\qquad p_\lambda=\frac{\mathrm d\mathbb{P}_\lambda}{\mathrm dQ}.
\]
The density exists by \citet[Theorem~3]{Leskela2024}. Let
$\mathbb{P}_0=\mathbb{P}_{\lambda_0}$ and $p_0=p_{\lambda_0}$ for the true intensity
$\lambda_0$. The intensity Hellinger distance is
\[
d_H(\lambda,\lambda_0)
=\left\{\int_D(\sqrt\lambda-\sqrt{\lambda_0})^2\,\mathrm d\mu\right\}^{1/2}.
\]
Let $X_1,X_2,\ldots\sim \mathbb{P}_0$ be independent copies of the process $X$, and set
\[
X^{(k)}=(X_1,\ldots,X_k),\qquad \mathbb{P}_\lambda^{(k)}=\mathbb{P}_\lambda^{\otimes k},\qquad \mathbb{P}_0^{(k)}=\mathbb{P}_{\lambda_0}^{(k)},\qquad \mathbb{P}_0^\infty=\mathbb{P}_0^{\otimes\mathbb N}.
\]
Write $\mathbb{E}_\lambda^{(k)}$ for expectation under $\mathbb{P}_\lambda^{(k)}$,
and $\mathbb{E}_0^{(k)}$ when $\lambda=\lambda_0$.
This is the iid framework of \citet[Ch.~6.4]{Ghosal2017}.
For $B\in\mathcal B(\Theta)$ and $S_k=\bigcup_{i=1}^kX_i$,
\begin{equation}\label{eqn: posterior}
\Pi(B\mid X^{(k)})
=\frac{\int_B\prod_{i=1}^kp_\lambda(X_i)\,\Pi(\mathrm d\lambda)}
       {\int_\Theta\prod_{i=1}^kp_\lambda(X_i)\,\Pi(\mathrm d\lambda)} =\Pi(B\mid S_k),
\end{equation}
$\mathbb{P}_0^\infty$-almost surely.

\begin{definition}[{\citet[][Definition 6.1]{Ghosal2017}}] \label{def: consistency}
Let $A_\epsilon:=\{\lambda\in\Theta: d_H(\lambda, \lambda_0) >\epsilon\}$. The posterior is said to be strongly consistent at $\lambda_0$ with respect to $d_H$ if $\mathbb{P}_0^\infty\{\forall \epsilon>0: \lim_{k\to\infty} \Pi(A_\epsilon\mid X^{(k)})=0\}=1$. 
\end{definition}
This posterior consistency can be verified via three sufficient conditions with the framework laid out by \citet{Ghosal2017} which are given as follows.

\subsection{Sufficient Conditions}
Denote the Kullback-Leibler (KL) divergence from $p_0$ to $p_\lambda$ by $K(p_0; p_\lambda):=\int p_0 \log (p_0/p_{\lambda}) \mathrm{d}Q = \int_D \{\lambda_0 \log(\lambda_0/\lambda) - \lambda_0 + \lambda\} \mu(\mathrm{d}\bfx)$, where the second equation due to Theorem 6 of \citet{Leskela2024}. The following three conditions are as follows. 
\begin{itemize}[nosep]
\item[\textbf{C1}.]  Prior mass condition. For every $\epsilon>0$, $\Pi\{\lambda: K(p_0; p_\lambda)<\epsilon\}>0$.
\item[\textbf{C2}.] Sieve condition. For each fixed $\epsilon>0$, there exist Borel sets $\calF_k\subset\Theta$ and a constant $c_1>0$ such that $\Pi(\calF_k^c)\le e^{-c_1k}$ for all sufficiently large $k$.
\item[\textbf{C3}.] Test condition. For the same $\epsilon$ and sieve sequence, there exist measurable tests $\psi_k=\psi_k(X^{(k)})\in[0,1]$ and constants $c_2,c_3>0$ such that
\begin{align*}
\mathbb{E}_0^{(k)}[\psi_k]&\le e^{-c_2k},
\sup_{\lambda\in\calF_k\cap A_\epsilon}
\mathbb{E}_\lambda^{(k)}[1-\psi_k]&\le e^{-c_3k}
\end{align*}
for all sufficiently large $k$. The sieves and constants in C2--C3 may depend on $\epsilon$.
\end{itemize}
If a posterior satisfies these three conditions, Theorem 6.17 of \citet{Ghosal2017} implies that the posterior is consistent in the sense of Definition~\ref{def: consistency}. Verification of these conditions for the soft PPT model is given in the Appendix. The assertion of the consistency of the posterior under these three conditions is also verified. Next, assumptions on the truth intensity function, model priors are stated explicitly. 

\subsection{Main Theoretical Result} 
\begin{assumption}\label{ass: truth}
Assume that the true intensity function is continuous, i.e., $\lambda_0\in \mathcal{C}(D)$, and is also bounded away from zero and $\infty$, i.e., $0<\lambda_-\leq \lambda_0(\bfx) \leq \lambda_+<\infty, \forall\bfx\in D$ for some constants $\lambda_-, \lambda_+$. Let $\mu$ be a diffuse measure with $\mu(D)\in (0, \infty)$. 
\end{assumption} 

\begin{assumption}\label{ass: tree prior}
Let $\alpha_0\in (0,1/2)$ be fixed. Under Bayesian CART prior, for any node $v$ at depth $h$, assume that $\rho(v) = \alpha_0^{h+1}$ for $h=0,1,\ldots$. $F_U((0,1)) =1$ and $F_U((1/2-e,1/2+e))>0$ for every $e\in (0,1/2)$. 
\end{assumption}
\begin{remark}
Note that the splitting probability assumption follows \citet{Rockova2019}, which is different from the employed prior in the model. Any of the following cut distributions $F_U$ satisfy the above assumption: (i) $F_U = \delta_{1/2}$, (ii) $F_U=\mathrm{Unif}(0,1)$, and (iii) $F_U=\sum_{i=1}^M w_i\delta_{u_i}$ where $u_i\in(0,1)$, $w_i>0$, $\sum_{i=1}^M w_i=1$, and $u_i=1/2$ for at least one $i$. Here $\delta_a$ denotes a point mass at $a$. 

\end{remark}

\begin{assumption} \label{ass: hyperparameters}
Let $\kappa, \xi\in(0,\infty)$ be fixed. Assume that $F_U((1/2-\epsilon, 1/2+\epsilon))>0$ for every $\epsilon\in (0, 1/2)$.  Suppose that the prior on $\gamma$ has a Lebesgue density on $(0,\infty)^d$ that has positive support on $(0,\infty)^d$. 
\end{assumption}

Under these assumptions, posterior consistency can be shown to hold true and is stated in the following result. 

\begin{theorem} \label{thm: posterior consistency}
If Assumptions~\ref{ass: truth},~\ref{ass: tree prior}, and~\ref{ass: hyperparameters} hold true, the posterior distribution~\eqref{eqn: posterior} is strongly consistent at $\lambda_0$ with respect to Hellinger distance $d_H$.
\end{theorem}

\begin{remark}
This result provides theoretical guarantee that the proposed soft-tree PPT model can estimate the intensity well under the proposed asymptotic framework. As illustrated in Section~\ref{sec: numerical}, intensity function under the S-PPT model can be well estimated. 
\end{remark}



\section{Numerical Illustrations} \label{sec: numerical}

The goal of numerical examples is to study its performance under both low-dimensional covariate space and high-dimensional covariate space. In low-dimensional space ($d\leq 2$), we can compare the proposed model with the widely-used LGCP since there is not available implementation for higher-dimensional space with $d>2$. In particular, we will use the INLA-approximated LGCP available  in the \textsf{R} package \texttt{inlabru} \citep{Lindgren2024}. For $d>2$, BART-PP \citep{Lamprinakou2023} and PPT \citep{Ma2026PPT} are closely related Bayesian competitors, however there is no available software code available for BART-PP. However, the code for XGBoostPP \citep{Lu2026} is available online and will be used for comparison in higher dimensions.   The PPT model is very similar to the proposed model except that PPT uses hard split for the model and sequential Monte Carlo (SMC) method for inference. It is implemented in the \texttt{R} package \texttt{poistree}\footnote{The \textsf{R} package \texttt{poistree} is available at \url{https://github.com/pulongma/poistree}}. 

To evaluate performance of the proposed model, we will mainly focus on comparing the estimated intensity function via scaled root mean square error (sRMSE) and log posterior predictive density (lppd) over testing datasets. Their formulas are given in Appendix~\ref{app: metrics}. sRMSE assesses the property of mean estimator of intensity function while lppd assess the predictive performance for the entire predictive distribution which is a desirable metric for evaluating Bayesian models. Smaller sRMSE indicates better estimation of intensity function while larger lppd indicates better performance of predictive distribution for the intensity. All the numerical experiments are performed on a 24GB memory Macbook Pro with Apple M4 Pro chip and 20 cores. All the timing results are reported in minutes. 
 
\subsection{2-D Example}

This example  is simulated from an inhomogeneous Poisson point process on the unit square \([0,1]^2\), with a smooth multi-scale spatially varying intensity
$
\lambda_0(\bfx)=C_0 f(\bfx)
$
with their specific formulas given in Appendix~\ref{app: 2D Example}. This example generates about 2500 random points using the Poisson thinning algorithm \citep[][p~11]{Moller2004}. 

Figure~\ref{fig: 2D example} compares the estimated intensity functions under LGCP, PPT, and S-PPT with the truth intensity function. As we can see from Figure~\ref{fig: 2D example}, both S-PPT and LGCP performs better than PPT and XGBoostPP, however S-PPT gives a much closer intensity surface to the truth. The INLA-LGCP approximates the latent GP with Mat\'ern-driven stochastic partial differential equation, which shows over-smooth behaviors. Both S-PPT and PPT show good performance in  capturing the sharp changes in the truth, however S-PPT allows smoother transitions than PPT. XGBoostPP shows some clearly sharp changes even though it is an ensemble tree method. This suggests that the Bayesian tree method with soft gating in S-PPT can better smooth over partition boundaries. Table~\ref{table: 2D example} confirms that S-PPT gives best performance in terms of sRMSE and lppd, although it takes much longer computing time than the other two methods. Here, S-PPT is run with 5,000 burn-in iterations and additional 10,000 iterations for two MCMC chains.

\begin{figure}[ht] 
\vspace{.3in}
\centering
\includegraphics[width=\linewidth]{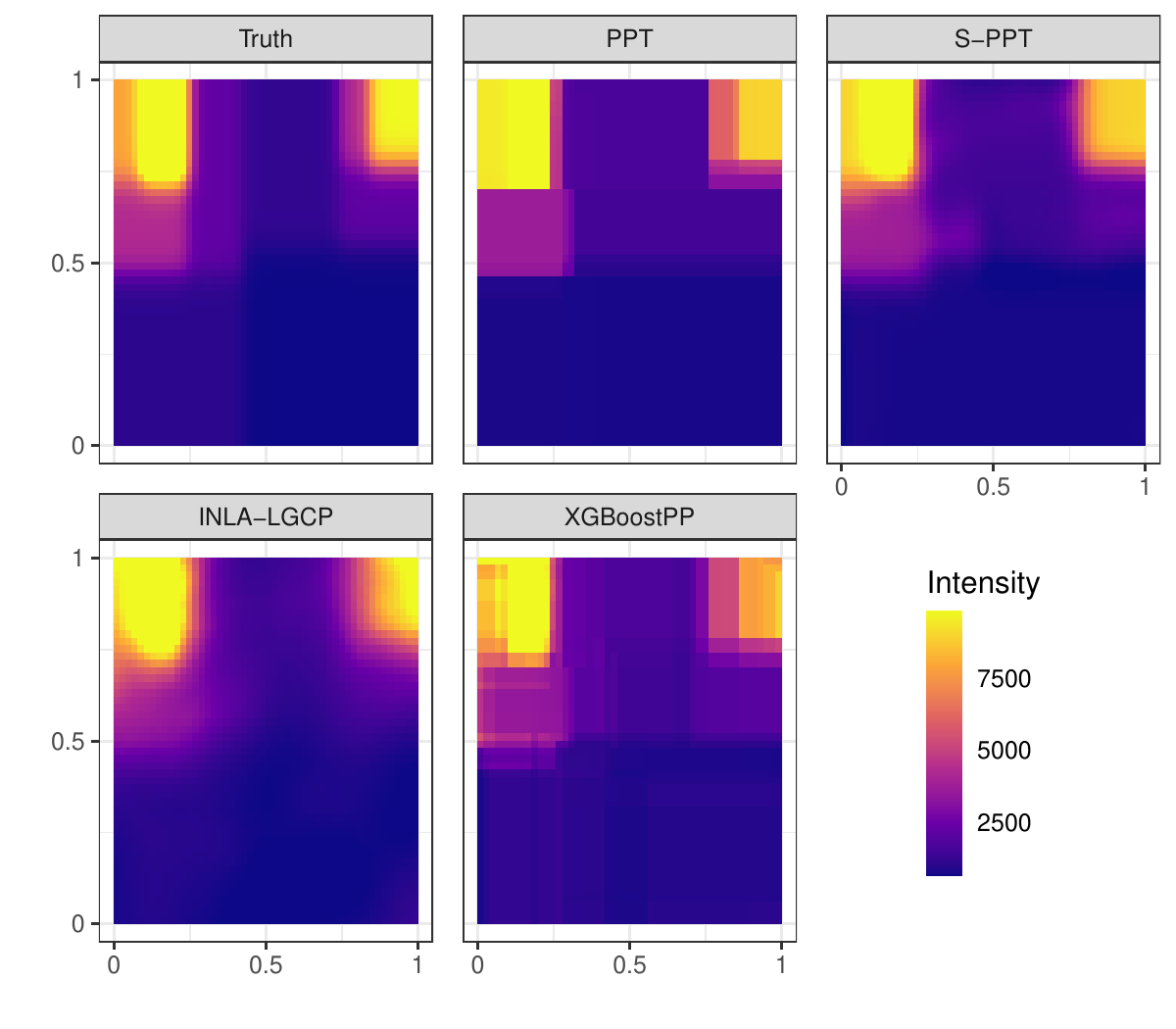}
\vspace{.3in}
\caption{Comparison of estimated intensities}
\label{fig: 2D example}
\end{figure}

\begin{table}[ht]
\caption{Performance in the 2-d example} \label{table: 2D example}
\begin{center}
\begin{tabular}{lll l}
\textbf{Model}  &\textbf{sRMSE} & \textbf{lppd} & Time\\ \hline \\
PPT         &0.290 & 2736 & 0.02\\
S-PPT             &\textbf{0.190} & \textbf{2742} & 3.77 \\
LGCP             &0.234 & 2738 & 0.68\\
XGBoostPP   & 0.233 & 2739 & 0.01 
\end{tabular}
\end{center}
\end{table}

\subsection{10-D Example}
For $\bfx=(x_1,\ldots,x_{10}) \in [0,1]^{10}$, the true intensity function is of the form 
$\lambda_0(\bfx) = n_0\frac{f_{10}(\bfx)}{Z},$
where $ Z=\int_Df_{10}(\bfx)\,d\bfx$ the normalizing constant and $n_0=15{,}000$ is expected number of points. The expression for $f_{10}(\cdot)$ is given in Appendix~\ref{app: 10d example}. combines a discontinuous three-level effect, two smooth oscillatory effects, a continuous cusp, a polynomial bump, and two bivariate terms with different regularity so that a nonlinear intensity with different degrees of smoothness imposed on different variables.  The variables $x_1,\ldots, x_8$ are active while  $x_9$ and $x_{10}$ are inactive. The dataset is simulated over the domain using thinning algorithm with $n_0$ expected number of points. The dataset is randomly split into $80/20\%$ for training and testing. 

Table~\ref{table: 10D example} shows that S-PPT performs best in terms of sRMSE and lppd. This example is encouraging since S-PPT is a single tree model and XGBoostPP is an ensemble tree model. One possible explanation is that its posterior distribution under S-PPT is an average of individual posterior distributions across posterior tree samples; as a result, any posterior estimates from S-PPT would behave somewhat like an estimate from an ensemble tree model.  This might justify the better performance of S-PPT over XGBoostPP in this particular example. Both S-PPT and XGBoostPP perform  better than the PPT model because  the underlying intensity favors a smoother intensity model unlike PPT.

\begin{table}[ht]
\caption{Performance in the 10-d example} \label{table: 10D example}
\begin{center}
\begin{tabular}{lll l}
\textbf{Model}  &\textbf{sRMSE} & \textbf{lppd} & Time\\
\hline \\
PPT         &0.221 & 21789 & 0.06\\
S-PPT             &\textbf{0.184} & \textbf{21818} & 131 \\
XGBoostPP   &0.213 & 21791 & 0.02\\
\end{tabular}
\end{center}
\end{table}

Figure~\ref{fig: 10D example}  also shows that S-PPT gives reasonable well estimates for the intensity function as the estimated intensity against true intensity aligns well along the 45 degree line. In Appendix~\ref{app: 10d example}, we also demonstrate how the relationship between intensity function and each variable is captured. In general, Figure~\ref{fig: 10D example marginal plot} shows both S-PPT and XGBoost can estimate several simple relationships well and S-PPT can correctly identify flat (or constant) relationship. For some marginal relationships such as $x_5$, both S-PPT and XGBoostPP did not capture it very well. 

\begin{figure}[ht] 
\vspace{.3in}
\centering
\includegraphics[width=\linewidth]{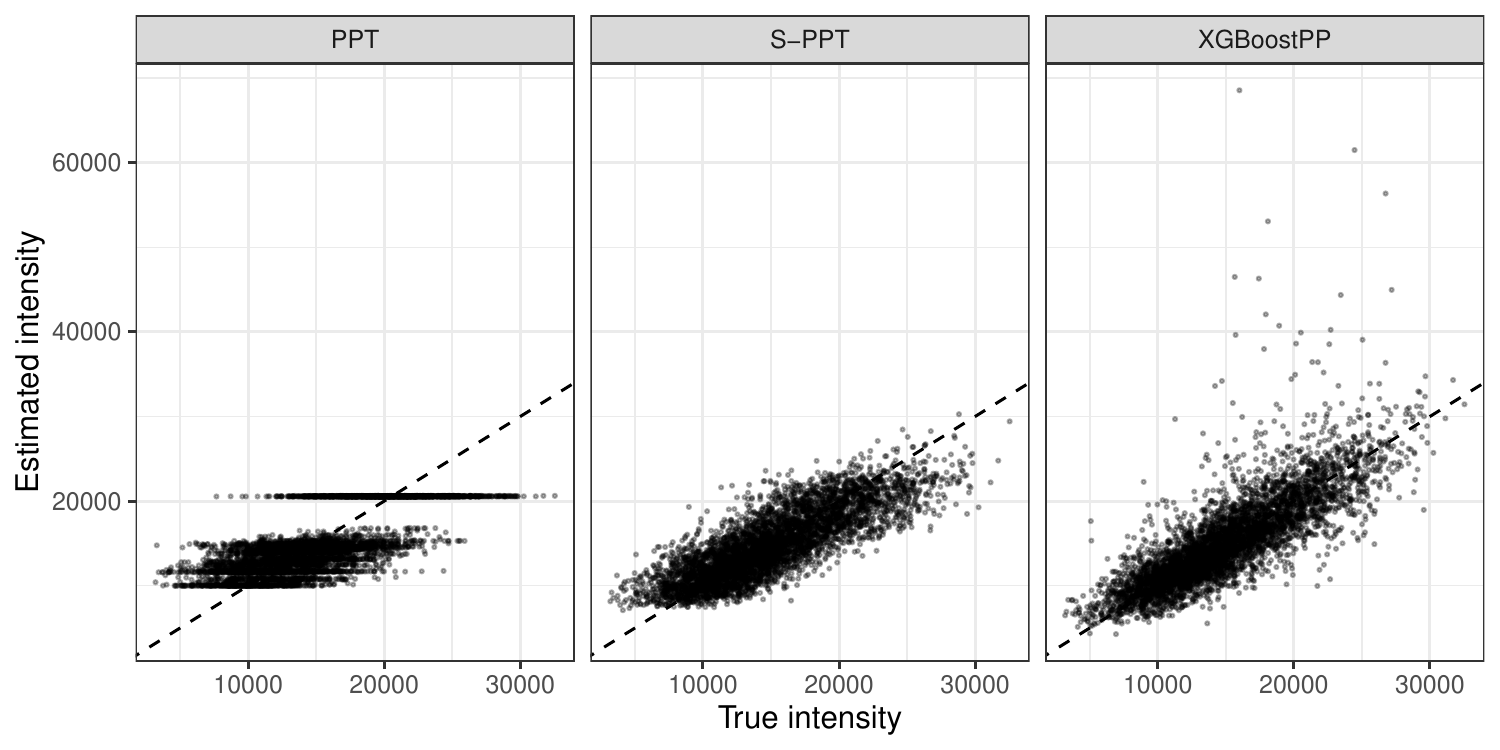}
\vspace{.3in}
\caption{10-d example scatterplot}
\label{fig: 10D example}
\end{figure}

\subsection{Real Data}

\paragraph{BCI forestry dataset:} This dataset refers to the classic tropical forest research data of tree species information collected in Barro Colorado Island (BCI), Panama which was used to study tree populations, forest dynamics, and biodiversity. This paper is focused on modeling spatial locations of a single tree species of Beilschmiedia pendula from the R dataset \texttt{bei} in the package \texttt{spatstat.data}. This dataset, originally analyzed by \citet{Moller2007}, consists of 3,604 spatial locations of tree species with 8 environmental covariates in a 1000m $\times$ 500m observation window: elevation (elev), magnitude of the elevation gradient (grad), Soil copper content (Cu), Nitrogen mineralization measure (Nmin), Soil phosphorus content (P), Soil acidity (pH), incoming solar radiation index (solar), topographical wetness index (twi). Computing time is measured in minutes. 

\paragraph{California wildfires dataset:}  There is an increasing interest in modeling wildfires and understanding how wildfire occurrence varies across space and seasons for different environmental conditions. We  consider 8,178  recorded wildfires events within the California study region during 2024 \citep{Short2026}. This dataset is obtained  from the seventh edition of the Fire Program Analysis Fire-Occurrence Database. This analysis is focused on nine covariates: easting, northing, elevation, monthly mean daily maximum temperature, monthly precipitation, monthly mean vapor-pressure deficit, monthly mean wind speed, and sine/cosine terms representing the annual seasonal cycle. These weather covariates are derived from gridMET data product.  These covariate information can be used to assess spatial and seasonal wildfire occurrence patterns and how this pattern was associated with weather and elevation  in California during 2024.   This dataset is then used to evaluate the performance under PPT, S-PPT, and XGBoostPP using a 80/20 split for training and testing. S-PPT is run with 5,000 burn-in iterations and additional 20,000 iterations for posterior inference.

\begin{table}[ht]
\caption{Performance of models on different datasets. Each dataset is randomly split into 80\% for training and 20\% for testing.} \label{table: real applications}
\begin{center}
\begin{tabular}{llll l}
Data    & \textbf{Model}  &\textbf{sRMSE} & \textbf{lppd} & Time    \\ \hline  \\
\multirow{3}{*}{Bei} & PPT         &0.340              & 4545      & 0.08    \\
                                & S-PPT      &\textbf{0.245}            & \textbf{4701}       & 6.00        \\
                                & XGBoostPP  & 0.264    & 4688            & 0.00    \vspace{12pt} \\  

\multirow{3}{*}{Wildfire}  & PPT         & 0.194  &  12192   &   2.52     \\            
                                & S-PPT     & \textbf{0.140}   &  12350   &  64.3           \\   
                     & XGBoostPP          &   0.173         &   \textbf{12583}    & 0.07       
\end{tabular}
\end{center}
\end{table}

The results in Table~\ref{table: real applications} also confirm that S-PPT model gives best sRMSE and lppd for Bei dataset but only give best sRMSE for wildfire data. One potential limitation of XGBoostPP compared to Bayesian models is that it does not provide confidence intervals for the intensity estimator but provide confidence intervals for counts; in contrast, Bayesian models provide a straightforward uncertainty assessment for any functional form of intensity function. Figure~\ref{fig: wildfire example count} shows the 18 cell counts with its confidence/credible intervals over a 5 by 5 cell regions over the study domain. For both S-PPT and XGBoostPP, their interval estimates look much better than that from PPT. In addition, Figure~\ref{fig: wildfire example} also shows the spatial map of estimated intensity functions. It is clear that this posterior mean of intensity function under S-PPT look much smoother than both XGBoostPP and PPT.   

\begin{figure}[!htbp] 
\vspace{.3in}
\centering
\includegraphics[width=\linewidth]{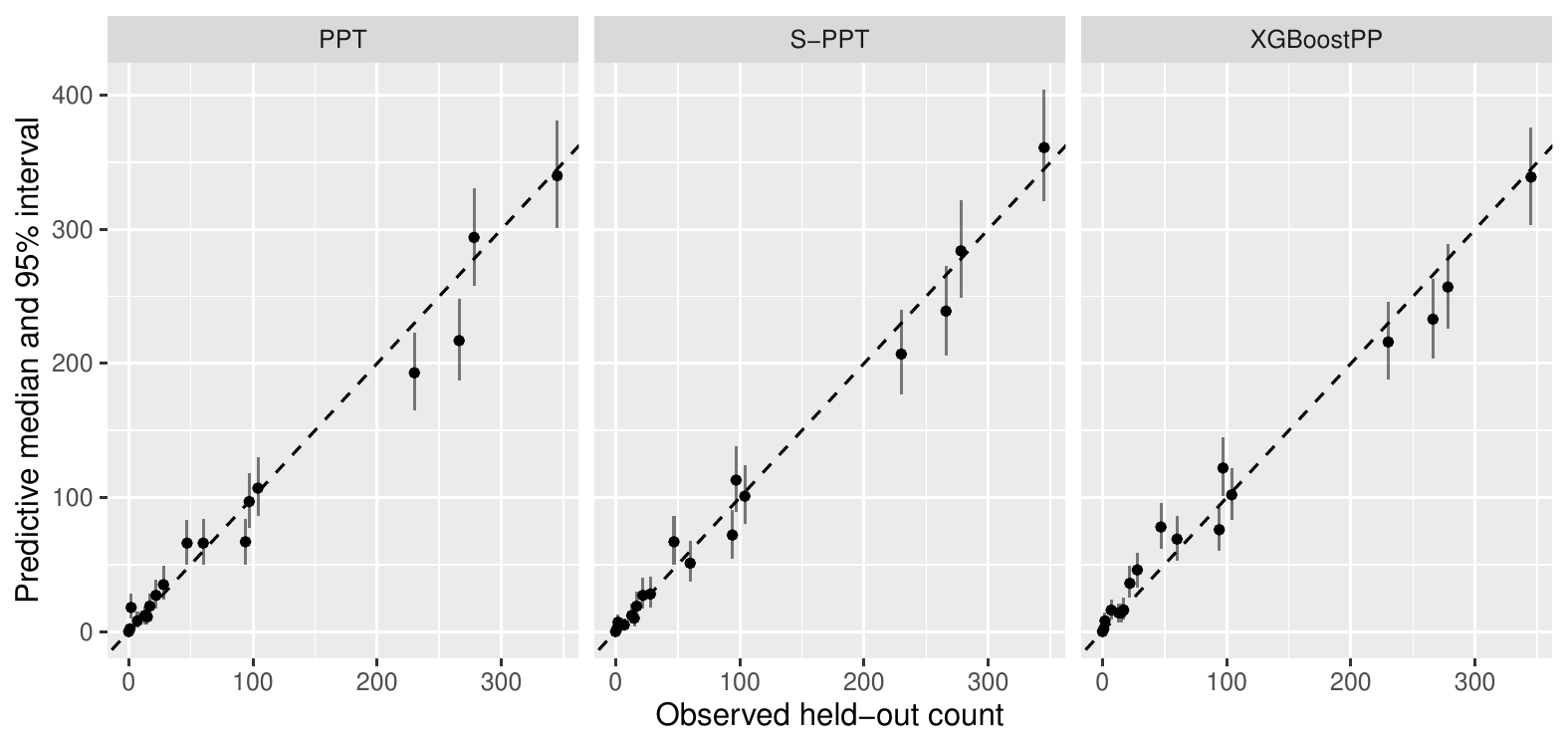}
\vspace{.3in}
\caption{Estimated count intensity function for California wildfire data under PPT, S-PPT, and XGBoostPP.}
\label{fig: wildfire example count}
\end{figure}

\section{Concluding Remarks}

This paper develops a Bayesian soft Poisson process tree (S-PPT) model for estimating intensity functions in a possibly high-dimensional covariate space. The proposed model exploits the Bayesian CART prior for learning tree partitions and integrates the soft-decision tree idea to smooth over partition boundaries. To carry out efficient Bayesian inference, this paper develops a data augmentation technique to allow closed-form conjugate sampling for intensity parameters in a partially collapsed Gibbs sampler. The theoretical result also shows strong consistency of the posterior distribution for estimating continuous intensity functions, thus providing good frequentist property of the proposed model with a slight modification on the tree prior.  

There are also several future research directions that could be pursued. First this model can be extended to an additive tree method to exploit the benefits of ensemble methods. Current computational strategy with closed-form random integral evaluation relies on axis-aligned partitions for rectangular regions. This might be relaxed if the domain is convex but not necessarily rectangular. Finally, the tree update with standard RJ-MCMC algorithm could be further improved using better proposals \citep[e.g.,][]{Pratola2016} for more efficient tree search or the sequential Monte Carlo method \citep{Ma2026PPT}.

\begin{singlespacing}
\bibliographystyle{apalike}
\bibliography{main}
\end{singlespacing}

\clearpage
\appendix

\setcounter{equation}{0}
\setcounter{page}{1}
\setcounter{table}{0}
\setcounter{section}{0}
\setcounter{figure}{0}
\renewcommand{\theequation}{S.\arabic{equation}}
\renewcommand{\thesection}{S.\arabic{section}}
\renewcommand{\thesubsection}{S.\arabic{section}.\arabic{subsection}}
\renewcommand{\thepage}{S.\arabic{page}}
\renewcommand{\thetable}{S.\arabic{table}}
\renewcommand{\thefigure}{S.\arabic{figure}}
\renewcommand{\theHfigure}{S.\arabic{figure}}

\clearpage\pagebreak\newpage

\renewcommand{\theequation}{S.\arabic{equation}}
\renewcommand{\thesubsection}{S.\arabic{section}.\arabic{subsection}}

\begin{center}
{\Large \textbf{Supplementary Material}} 
\end{center}

\section{Model Construction} \label{app: posterior computation} 

This section derives the analytic exposure integral and the allocation probabilities, and gives details of posterior computation.

\subsection{Proof of Proposition~\ref{prop: analytic exposure}}
\label{app: analytic exposure proof}
\begin{proof}
Observe that 
\[
I_v=\prod_{j=1}^d
\underbrace{\int_{a_j}^{b_j}
\prod_{\substack{u\in\downarrow\!v\\j_u=j}}
g_u(t)^{\delta_{uv}}\{1-g_u(t)\}^{1-\delta_{uv}}\,\mathrm dt}_{J_{vj}}.
\]
where an empty product reduce to $J_{vj}=W_j$. Applying a changes-of-variables technique yields 
\[
u=\exp\left\{\frac\gamma W\left(t-\frac{a_j+b_j}{2}\right)\right\},
\qquad \mathrm dt=\frac W\gamma\frac{\mathrm du}{u},
\qquad u_\pm=e^{\pm\gamma/2}.
\]
The right and left factors become $u/(u+q_k)$ and $q_k/(u+q_k)$.
Since the $q_k$ are distinct, direct calculation yields
\begin{equation}\label{eqn: rational exposure integral}
J_{vj}
=\frac W\gamma\int_{u_-}^{u_+}
\frac{C u^{r-1}}{\prod_{k=1}^m(u+q_k)}\,\mathrm du =\frac W\gamma\int_{u_-}^{u_+}
\left\{\frac{\mathds{1}_{\{r=0\}}}{u}
+\sum_{k=1}^m\frac{A_k}{u+q_k}\right\}\,\mathrm du.
\end{equation}
where 
\begin{align*}
A_k := \lim_{u \to - q_k} (u + q_k) R(u) = \lim_{u \to - q_k} \frac{Cu^{r-1}}{\prod_{\ell\neq k} (u + q_\ell)} = \frac{C(-q_k)^{r-1}}{\prod_{\ell\neq k} (q_\ell - q_k)}.
\end{align*}
Note that 
\begin{align*}
\int_{u_-}^{u_+} \frac{A_k}{u + q_k} \mathrm{d} u = A_k \log \frac{u_+q_k}{u_- + q_k}. 
\end{align*} 
Thus it follows from direct calculation and $\log(u_+/u_-)=\gamma$ that 
 \[
J_{vj}=\frac W\gamma\left\{
\gamma\mathds{1}_{\{r=0\}}+
\sum_{k=1}^m A_k\log
\frac{e^{\gamma/2}+q_k}{e^{-\gamma/2}+q_k}
\right\},
\]
as desired. 
\end{proof}

\subsection{Proof of Proposition~\ref{prop: color probability}}
\label{app: color probability proof}
\begin{proof}
It follows from Eq.~\eqref{eqn: integrated likelihood} that  terms independent of $v$ would be canceled  and thus 
\[
\Pr(z_i=v\mid z_{-i},\bfx_{1:n},T,\gamma)
\propto\phi_v(\bfx_i)\frac{h(n_v^{-i}+1,I_v)}{h(n_v^{-i},I_v)}
=\phi_v(\bfx_i)\frac{\kappa+n_v^{-i}}{\kappa/\xi+I_v}.\qedhere
\]
\end{proof}

\subsection{Tree updates}\label{app: PCGS}

Let $\mathcal G(T)$ be the splittable leaves and
$\mathcal C(T)$ the nodes with two leaf children. Let $q_v(r)$ be the
normalized admissible rule prior for $r=(j,u)$, determined by the
parent cell and fixed covariates, independently of its current rule
and labels. Define
\[
H(T,z)=\prod_{v\in\mathcal L(T)}h(n_v,I_v),\qquad
s_i(r)=\frac{\phi_{v_\lchild}(\bfx_i;r)}{\phi_v(\bfx_i)},
\]
using the complete $h$ in Eq.~\eqref{eqn: integrated likelihood};
$s_i(r)$ is evaluated under split $r$.
Choose
\[
M\sim\mathrm{Unif}\{\mathrm{GROW},\mathrm{PRUNE},\mathrm{CHANGE}\}.
\]
If the selected move has no eligible node,  $(T,z)$ is retained.

\begin{enumerate}[leftmargin=*,nosep]
\item \textbf{GROW: split one leaf.}
\begin{align*}
v&\sim\mathrm{Unif}\{\mathcal G(T)\},\qquad r^*\sim q_v,\qquad
T^*=\mathrm{Grow}(T,v,r^*),\\
z_i^*&\sim\mathrm{Cat}_{(v_\lchild,v_\rchild)}
\{s_i(r^*),1-s_i(r^*)\},\qquad i:z_i=v.
\end{align*}
\item \textbf{PRUNE: merge two leaf children.}
\[
v\sim\mathrm{Unif}\{\mathcal C(T)\},\qquad
T^*=\mathrm{Prune}(T,v),\qquad z_i^*=v,\qquad i:z_i\in\{v_\lchild,v_\rchild\}.
\]
\item \textbf{CHANGE: replace one split rule.}
\begin{align*}
v&\sim\mathrm{Unif}\{\mathcal C(T)\},\qquad r^*\sim q_v,\qquad
T^*=\mathrm{Change}(T,v,r^*),\\
z_i^*&\sim\mathrm{Cat}_{(v_\lchild,v_\rchild)}
\{s_i(r^*),1-s_i(r^*)\},\qquad
 i:z_i\in\{v_\lchild,v_\rchild\}.
\end{align*}
\end{enumerate}

The routing factors cancel against the label-proposal probabilities.
A proposed value  $(T^*,z^*)$ is accepted with probability
\[
\alpha=\min\left\{1,
\frac{\pi(T^*)H(T^*,z^*)}{\pi(T)H(T,z)}R_M\right\},
\]
where 
$
R_M=\begin{cases}
\dfrac{|\mathcal G(T)|}{|\mathcal C(T^*)|q_v(r^*)},&M=\mathrm{GROW},\\
\dfrac{|\mathcal C(T)|q_v(r)}{|\mathcal G(T^*)|},&M=\mathrm{PRUNE},\\
\dfrac{q_v(r)}{q_v(r^*)},&M=\mathrm{CHANGE}.
\end{cases}
$

Here $r$ is the current or removed rule. On rejection,  $(T,z)$ is retained.

\subsection{Gate update}\label{app: gate update}

Fix $T,\bflambda$ and set $\eta=\log\gamma$ componentwise. With the
shape--rate Gamma prior, the log target, including the Jacobian, is
\[
\ell(\eta)=\log p(\bfx_{1:n}\mid\bflambda,e^\eta,T)
            +\sum_{j=1}^d(a_0\eta_j-b_0e^{\eta_j}).
\]
At iteration $t=1,2,\ldots$, propose and accept with probability
\[
u\sim N_d(0,I_d),\qquad \eta^*=\eta+L_tu,\qquad
\alpha_{\gamma,t}=\min\{1,e^{\ell(\eta^*)-\ell(\eta)}\},
\]
where $L_tL_t^\top$ is the proposal covariance. On acceptance, set
$\eta=\eta^*$; otherwise retain $\eta$. Set $\gamma=e^\eta$ and
restore $z$ as in the main scan.
The stated prior permits every $\gamma_j>0$. Thus, any additional lower
bound is a prior truncation and proposals outside it are rejected.

During warm-up, robust adaptive Metropolis \citep{Vihola2012} uses
\[
L_{t+1}L_{t+1}^\top
=L_t\left[I_d+\omega_t(\alpha_{\gamma,t}-0.234)
                  \frac{uu^\top}{\|u\|^2}\right]L_t^\top,
\qquad \omega_t=\min\{1,dt^{-0.7}\}.
\]
Initialize $L_1=\operatorname{diag}(s_1,\ldots,s_d)$ with proposal
standard deviations $s_j>0$. Freeze $L_t$ after
$t_{\rm adapt}\le t_{\rm burn}$;
the implementation default is
$t_{\rm adapt}=\lfloor0.8t_{\rm burn}\rfloor$.
If a Cholesky update fails numerically,  the previous factor is retained.

\section{Numerical Examples}

\subsection{Performance Metrics} \label{app: metrics} 
Let $\bfx_{\mathrm{train}}$ and $\bfx_{\mathrm{test}}$ denote the observed training and test patterns. Let $\bfx_1^*,\ldots,\bfx_M^*\in D$ be evaluation sites and $\theta=(T,\bflambda,\gamma)$. Write $\widehat\lambda$ for the posterior mean intensity. The metrics are 
\[
\operatorname{sRMSE}
=
\frac{
\sqrt{\frac{1}{M}\sum_{m=1}^{M}
\left\{\widehat{\lambda}(\bfx_m^*)-\lambda_0(\bfx_m^*)\right\}^{2}}
}{
\frac{1}{M}\sum_{m=1}^{M}\lambda_0(\bfx_m^*)
}.
\]

\[
\begin{aligned}
\operatorname{lppd}
&=
\log \int
p(\bfx_{\mathrm{test}}\mid\theta)\,
p(\theta\mid \bfx_{\mathrm{train}})\,\mathrm{d}\theta \\
&=
\log \int
\exp\!\left\{
\sum_{\bfx\in \bfx_{\mathrm{test}}}\log\lambda_\theta(\bfx)
-
\int_D\lambda_\theta(\bfx)\, \mathrm{d}\bfx
\right\}
p(\theta\mid \bfx_{\mathrm{train}})\,\mathrm{d}\theta .
\end{aligned}
\]

\subsection{2-d Example} \label{app: 2D Example} 

The true intensity is of the form
$
{\lambda_0(\bfx)=C_0 f(\bfx),\, \bfx\in[0,1]^2.}
$
where $f(\cdot)$ is a multi-scale function varying in space and $C_0$ is chosen to normalize the intensity so that the expected number of counts is 2500 in the simulation example. 

In particular, the function $f$ is chosen according to the gating function structures at multiple scales of the form 
\begin{align*}
f(\bfx) &=b+\sum_{s=1}^{2}\sum_{k=1}^{5}
       (\beta_{s,k}-\beta_{s,k-1})\prod_{h=1}^{k}\ell_{s,h}(\bfx)
\end{align*} 
where   
$
\ell_{s,k}(\bfx)
 =\left[1+\exp\{-(2\delta_{s,k}-1)\gamma_{j_k}
                    (x_{j_k}-c_{s,k})\}\right]^{-1}
$, $\delta_{s,k}=0$ or $1$ selects the left or right child, and $j_k, c_{s,k}$ are cutting axes and cut locations.  $(\gamma_1,\gamma_2)=(78,48)$. Then $C_0$ is chosen as  
\[
C_0=\frac{n_0}{\displaystyle m^{-2}
 \sum_{i=1}^{m}\sum_{j=1}^{m}
 f\!\left(\frac{i-1/2}{m},\frac{j-1/2}{m}\right)},
\qquad n_0=2500,\quad m=500,
\]
so that $\int_D\lambda_0(\bfx)\,\mathrm d\bfx\approx2500$.

This function $f$ represents a multiscale structure through $w_{s,k}:=\beta_{s,k}-\beta_{s,k-1}$ and tree partitions, where 
\begin{gather*}
\beta_{1,0}=\beta_{2,0}=0.3,\\
(\beta_{1,1},\ldots,\beta_{1,5})=(0.80,1.70,3.25,5.95,11.15),\\
(\beta_{2,1},\ldots,\beta_{2,5})=(0.50,0.90,1.65,3.20,7.40).
\end{gather*}

For the tree partitions, the cut axis is chosen as $j_k=1+((k-1)\bmod2), k=1,\ldots, 5$. Let $A_{s,0}=D=[0,1]^2$ and $A_{s,k-1}=\prod_{j=1}^2 [a_{s,k-1,j}, b_{s, k-1,j}]$. The recursive formulas for cut location $c_{s,k}$ and $A_{s,k}$ are given by
\[
c_{s,k}
 =a_{s,k-1,j_k}+U_{s,k}(b_{s,k-1,j_k}-a_{s,k-1,j_k}),\qquad A_{s,k}
 =A_{s,k-1}\cap
 \{\bfx:(2\delta_{s,k}-1)(x_{j_k}-c_{s,k})\ge0\}
\]
where 
\begin{gather*}
(u_1,\ldots,u_5)=(0.42,0.58,0.47,0.55,0.38),\\
(U_{1,1},\ldots,U_{1,5})=(u_1,u_3,u_2,u_3,u_5),\qquad
(\delta_{1,1},\ldots,\delta_{1,5})=(0,1,0,1,1),\\
(U_{2,1},\ldots,U_{2,5})=(u_1,u_4,u_4,u_3,u_5),\qquad
(\delta_{2,1},\ldots,\delta_{2,5})=(1,1,1,1,1),
\end{gather*}

\begin{figure}[ht] 
\vspace{.3in}
\centering
\includegraphics[width=\linewidth, height=0.25\textheight]{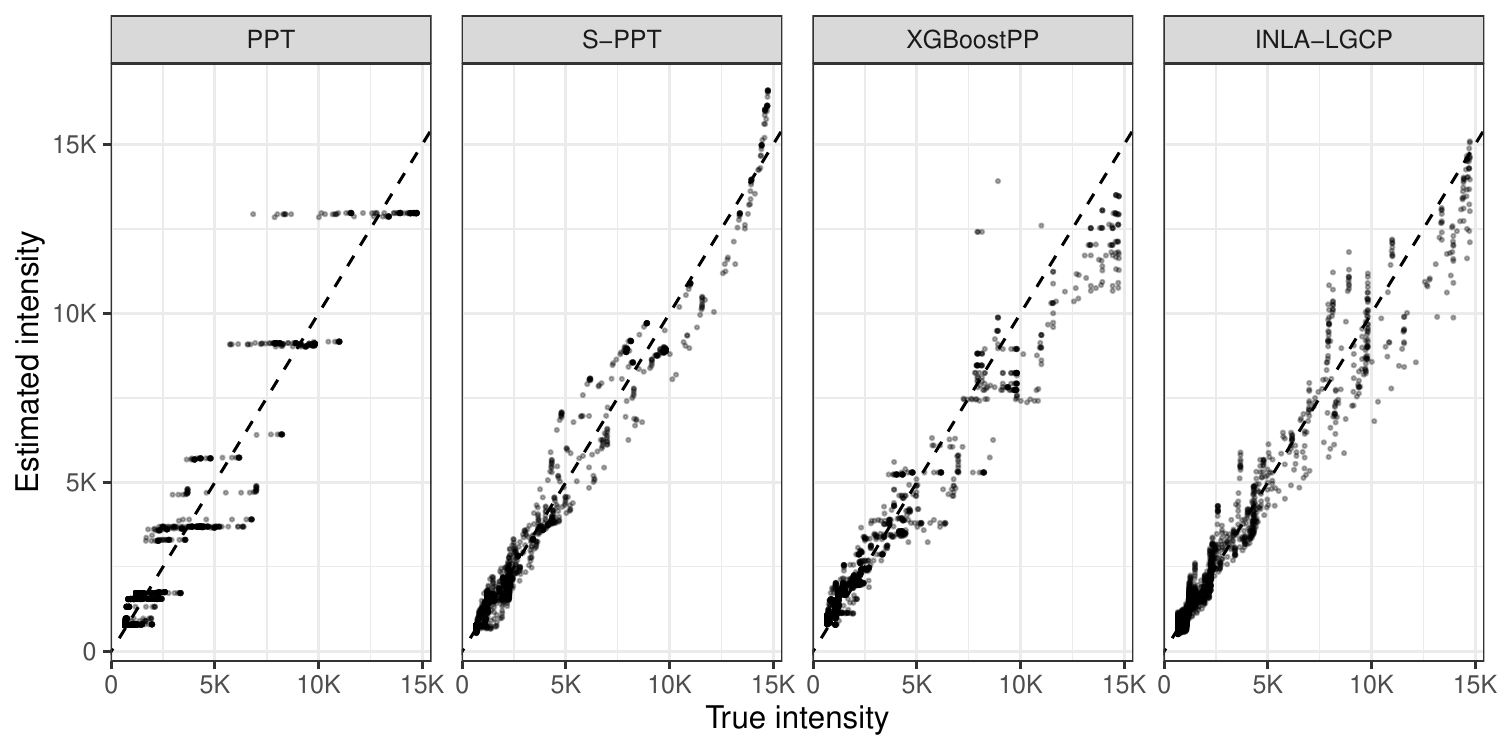}
\vspace{.3in}
\caption{2-d example scatterplot}
\label{fig: 2D scatter}
\end{figure}

\newpage
\subsection{10-d Example} \label{app: 10d example}

For $\bfx=(x_1,\ldots,x_{10}) \in [0,1]^{10}$, the true intensity function is of the form 
$
  \lambda_0(\bfx) = n_0\frac{f_{10}(\bfx)}{Z}, 
$
where $ Z=\int_Df_{10}(\bfx)\,d\bfx$, $n_0=15{,}000$, and 
$
  f_{10}(\bfx) =0.15+0.3\sum_{j=1}^{5}q_j(x_j) +0.6q_{67}(x_6,x_7)+0.6q_{78}(x_7,x_8).
$

The component functions are
\begin{align*}
q_1(u)
  &=\begin{cases}
       0,    &0\leq u<0.25,\\
       1,    &0.25\leq u<0.65,\\
       0.45, &0.65\leq u\leq1,
     \end{cases}
     \\[0.3em]
q_2(u)&=\frac{1+\sin(2\pi u)}{2},
  \qquad q_3(u)=\frac{1+\cos(3\pi u)}{2},
    \\
q_4(u)&=\frac{|u-0.3|}{0.7},
  \qquad q_5(u)=16u^2(1-u)^2,
    \\
q_{67}(u,v)&=\mathds{1}_{\{u\geq1/2\}}\mathds{1}_{\{v\geq1/2\}},
  \qquad q_{78}(u,v)=\sin^2(\pi u)\sin^2(\pi v),
\end{align*}
for $u,v\in[0,1]$.

\begin{figure}[ht] 
\vspace{.3in}
\centering
\includegraphics[width=\linewidth]{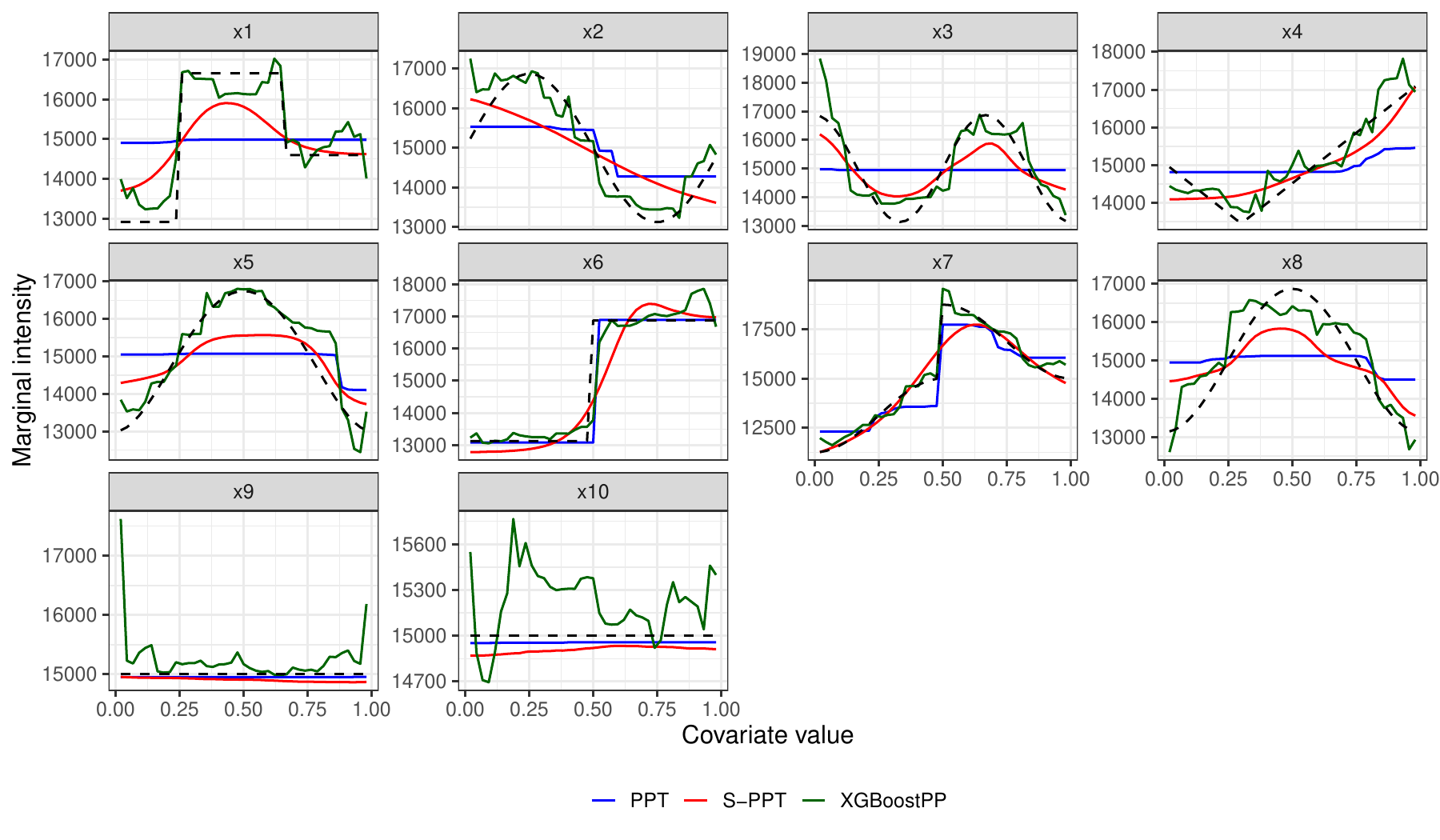}
\vspace{.3in}
\caption{Estimation of marginal relationships between intensity function and each variable. }
\label{fig: 10D example marginal plot}
\end{figure}

\newpage
\subsection{Real Application}

\begin{figure}[!htbp] 
\vspace{.3in}
\centering
\includegraphics[width=\textwidth, height=0.4\textheight]{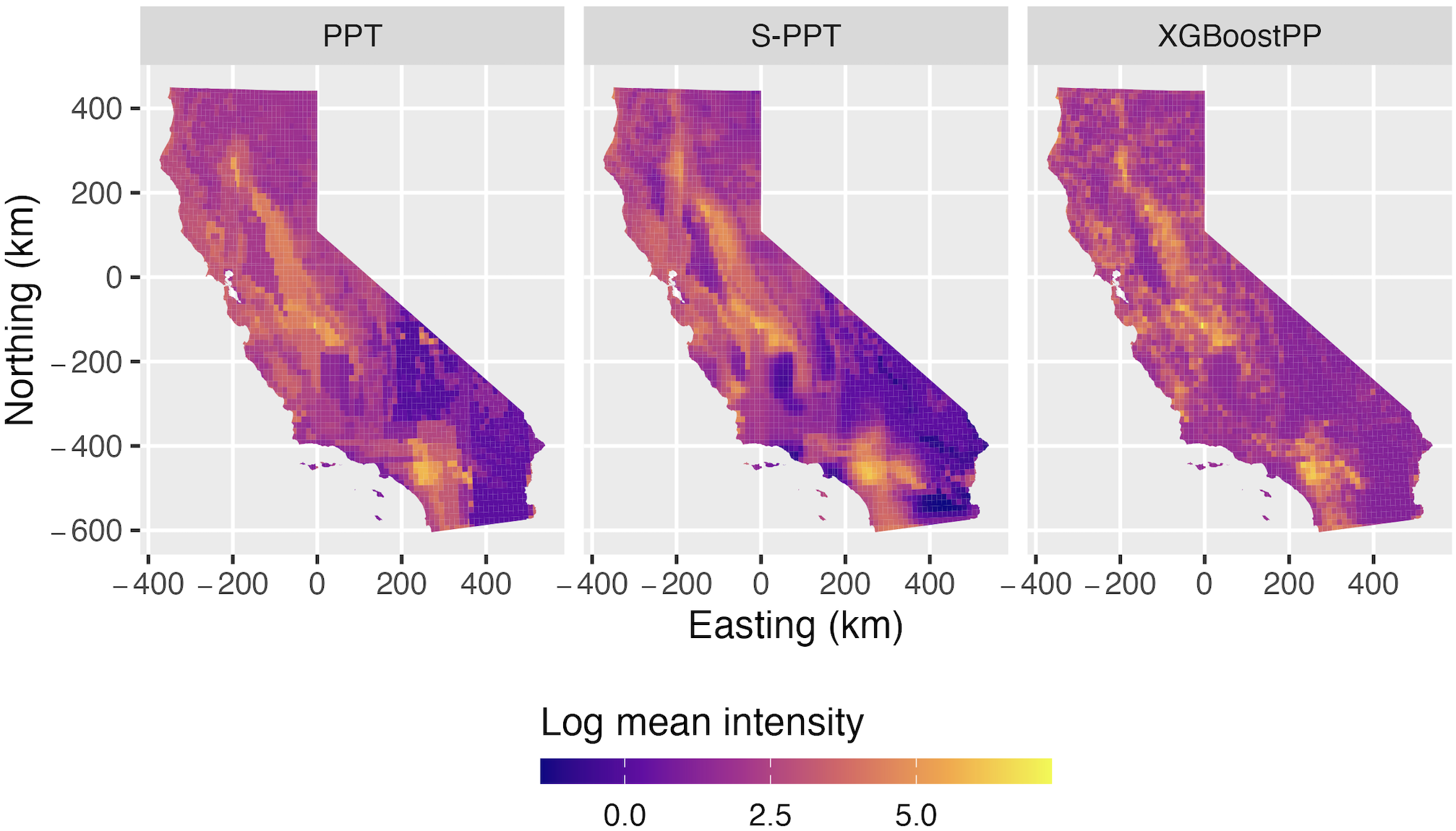}
\vspace{.3in}
\caption{Spatial map of log of estimated intensity function for California wildfire data under PPT, S-PPT, and XGBoostPP.}
\label{fig: wildfire example}
\end{figure}
\pagebreak

\section{Proof of Posterior Consistency}
\subsection{Verification of Condition C1}\label{app: C1}

As in Section~\ref{sec: theory}, let $D=[0,1]^d$ and write $V=\mu(D)$.
For functions, $\|\cdot\|_p$ denotes the $L^p(\mu)$ norm, including
the essential supremum for $p=\infty$; for vectors, $\|\cdot\|_\infty$
denotes the maximum norm.

\begin{lemma}[Kullback--Leibler support]\label{lem: KL support}
Under Assumptions~\ref{ass: truth}, \ref{ass: tree prior}, and
\ref{ass: hyperparameters},
\[
 \Pi\{\lambda:K(p_0;p_\lambda)<\epsilon\}>0
 \qquad(\epsilon>0).
\]
\end{lemma}

\begin{proof}
We construct a finite-tree approximation and a parameter neighborhood
of positive prior probability that preserves it; see
\citet[Example~7.6]{Ghosal2017} and
\citet[Appendix~B, (8)]{Linero2018}. The approximation below is specific
to the S-PPT gates.

\medskip
\noindent\emph{Step 1. Reduce KL support to uniform approximation.}
Fix $\epsilon>0$ and choose
\[
 0<e<\min\left\{\frac{\lambda_-}{2},
                   \sqrt{\frac{\epsilon\lambda_-}{2V}}\right\}.
\]
If $\|\lambda-\lambda_0\|_\infty<e$, then
$\lambda>\lambda_--e>\lambda_-/2$ $\mu$-almost everywhere.

Since $(V\lambda_0)(\mu/V)=\lambda_0\mu$ and
$(V\lambda)(\mu/V)=\lambda\mu$, An application of Appendix~A of \citet{Gugushvili2013} with reference measure $\mu/V$ and intensities $V\lambda_0$ and $V\lambda$ yields that 
\begin{equation}\label{eqn: uniform approximation}
 K(p_0;p_\lambda)
 \le\int_D\frac{(\lambda_0-\lambda)^2}{\lambda}\,\mathrm d\mu
 \le\frac{V\|\lambda-\lambda_0\|_\infty^2}{\lambda_--e}
 <\frac{2Ve^2}{\lambda_-}<\epsilon,
\end{equation}
where the remaining inequalities follow from
$\|\lambda-\lambda_0\|_\infty<e$, $\lambda\ge\lambda_--e$,
and the choice of $e$.

\medskip
\noindent\emph{Step 2. Construct a finite-tree approximation.}
Let
\[
 \omega_0(t)=
 \sup_{\substack{\bfx,\bfy\in D\\
                 \|\bfx-\bfy\|_\infty\le t}}
 |\lambda_0(\bfx)-\lambda_0(\bfy)|.
\]
It follows from the continuity of $\lambda_0$ on compact $D$ that $\omega_0(t)\to0$
as $t\downarrow0$.

\smallskip
\noindent\emph{(a) Tree and leaf values.}
For $r\ge1$, let $T_r$ be the complete binary tree of depth $dr$, with
\[
 j_u=1+(\operatorname{depth}(u)\bmod d),\qquad U_u^0=1/2,
 \qquad u\in\mathcal I(T_r).
\]
Write the box at node $v$ as
$\overline A_v^0=\prod_j[a_{vj}^0,b_{vj}^0]$ and set
$h=2^{-r}$, $J=2^{dr}$. Every leaf path has $r$ midpoint splits per
coordinate, so there are $J$ leaves, each with side length $h$.
Choose $\bfx_v\in\overline A_v^0$ for each leaf and set
$\lambda_v^0=\lambda_0(\bfx_v)$ and $\gamma_j^0=G>0$.
For $u\in\downarrow\!v$, let $\delta_{uv}$ be the right-turn indicator defined above. Then
\begin{gather*}
 c_u^0=(a_{u,j_u}^0+b_{u,j_u}^0)/2,\qquad
 \phi_v^0(\bfx)=
 \prod_{u\in\downarrow\!v}
 \left\{1+\exp\left[-(2\delta_{uv}-1)G(x_{j_u}-c_u^0)\right]\right\}^{-1},\\
 \sum_{v\in\mathcal L(T_r)}\phi_v^0=1,\qquad
 \lambda_{r,G}
 =\sum_{v\in\mathcal L(T_r)}\lambda_0(\bfx_v)\phi_v^0.
\end{gather*}

\smallskip
\noindent\emph{(b) Distant representatives.}
The ancestral cuts give
\begin{align*}
 a_{vj}^0
 &=\max\bigl(\{0\}\cup
     \{c_u^0:u\in\downarrow\!v,\ j_u=j,\ \delta_{uv}=1\}\bigr),\\
 b_{vj}^0
 &=\min\bigl(\{1\}\cup
     \{c_u^0:u\in\downarrow\!v,\ j_u=j,\ \delta_{uv}=0\}\bigr).
\end{align*}
If $\|\bfx-\bfx_v\|_\infty>2h$, some coordinate $j$ satisfies
\[
 x_j<x_{vj}-2h\le a_{vj}^0-h
 \quad\text{or}\quad
 x_j>x_{vj}+2h\ge b_{vj}^0+h.
\]
The relevant endpoint lies inside $(0,1)$ and is an ancestral cut.
Thus some $u\in\downarrow\!v$ satisfies
$(2\delta_{uv}-1)(x_{j_u}-c_u^0)<-h$. Since all routing factors lie in $[0,1]$,
\[
 \phi_v^0(\bfx)
 \le\exp\{(2\delta_{uv}-1)G(x_{j_u}-c_u^0)\}
 \le e^{-Gh}.
\]

\smallskip
\noindent\emph{(c) Total approximation error.}
By Assumption~\ref{ass: truth} and the partition of unity,
\begin{align*}
 |\lambda_{r,G}(\bfx)-\lambda_0(\bfx)|
 &\le\sum_v\phi_v^0(\bfx)
       |\lambda_0(\bfx_v)-\lambda_0(\bfx)|\\
 &\le\omega_0(2h)+2\lambda_+
       \sum_{\|\bfx-\bfx_v\|_\infty>2h}\phi_v^0(\bfx)\\
 &\le\omega_0(2h)+2\lambda_+Je^{-Gh}.
\end{align*}
Choose $r$ with $\omega_0(2^{1-r})<e/4$ and
$G>\max\{1,h^{-1}\log(8\lambda_+J/e)\}$. Then
\[
 \sup_{\bfx\in D}
 |\lambda_{r,G}(\bfx)-\lambda_0(\bfx)|<e/2.
\]

\medskip
\noindent\emph{Step 3. Preserve the approximation in a parameter neighborhood.}
Keep $T_r$ and its axes fixed, and write
\[
 \theta=
 \bigl((U_v)_{v\in\mathcal I(T_r)},\bflambda,\gamma\bigr),\qquad
 \theta^0=((1/2)_v,(\lambda_v^0)_v,(G)_j).
\]
Choose $0<\delta_0<\min\{1/4,G/2,\lambda_-/2\}$ and let
\[
 \mathcal K=
 \{\theta:\|\theta-\theta^0\|_\infty\le\delta_0\}.
\]
On $\mathcal K$, every cut fraction lies in $[1/4,3/4]$, so cuts
remain strictly inside their node intervals. The recursion
\[
 c_v=(1-U_v)a_{v,j_v}+U_vb_{v,j_v}
\]
ensures the continuity of the endpoints and cuts in the cut fractions.
The root-width gate predictor $\gamma_{j_v}(x_{j_v}-c_v)$ is therefore
continuous in $(\theta,\bfx)$. Hence $(\theta,\bfx)\mapsto\lambda_\theta(\bfx)$
is uniformly continuous on $\mathcal K\times D$.

Choose $0<\delta<\delta_0$ such that
\[
 \sup_{\|\theta-\theta^0\|_\infty<\delta}
 \sup_{\bfx\in D}
 |\lambda_\theta(\bfx)-\lambda_{r,G}(\bfx)|<e/2.
\]

It follows from the triangle inequality and Steps~1--2 that
\[
 \|\theta-\theta^0\|_\infty<\delta
 \quad\Longrightarrow\quad
 K(p_0;p_{\lambda_\theta})<\epsilon.
\]

\medskip
\noindent\emph{Step 4. Verify positive prior probability.}
Let $\mathcal U_\delta$ fix the selected topology and axes and impose
$\|\theta-\theta^0\|_\infty<\delta$.
Let $f_*$ denote the $\mathrm{Ga}(\kappa,\kappa/\xi)$ density. Then
\begin{align*}
 \pi(\mathcal U_\delta)
 &=\prod_{u\in\mathcal I(T_r)}\rho(u)\zeta_{j_u}
   \prod_{v\in\mathcal L(T_r)}(1-\rho(v))
   F_U((1/2-\delta,1/2+\delta))^{J-1}\\
 &\quad\times
   \int_{(G-\delta,G+\delta)^d}\pi(\gamma)\,\mathrm d\gamma
   \prod_{v\in\mathcal L(T_r)}
   \int_{\lambda_v^0-\delta}^{\lambda_v^0+\delta}
                         f_*(t)\,\mathrm dt
 >0.
\end{align*}
By Assumption~\ref{ass: tree prior}, the tree and cut factors are positive.
By Assumption~\ref{ass: hyperparameters}, the integrals are positive
because their densities are continuous and strictly positive on the
corresponding compact boxes. Therefore,
\[
 0<\pi(\mathcal U_\delta)
 \le\Pi\{\lambda:K(p_0;p_\lambda)<\epsilon\}.
\]
\end{proof}

\subsection{Sieve construction}\label{app: sieve construction}

We specify one sieve sequence and verify C2--C3 by following the
sieve and finite-cover testing strategy in
\citet[proof of Theorem~6.23]{Ghosal2017}; see also
\citet[Section~7]{Rockova2019} and
\citet[Supplement, Section~S.2]{Linero2018}.

Define:
\[
 G_{c,\tau}(y)=\{1+e^{-(y-c)/\tau}\}^{-1},\qquad
 g_v(\bfx)=G_{c_v,\tau_v}(x_{j_v}),\qquad
 c_v\in[0,1],\quad\tau_v>0.
\]
For a topology $T$ with $L$ leaves and axes $\boldsymbol j$, let
$\phi_v^g$ denote the routing weights determined by its gates $g$.
For a leaf-count bound $J\in\mathbb N$ and a leaf-intensity bound
$B>0$, define
\[
\mathcal E_{J,B}
 =\bigcup_{(T,\boldsymbol j):\,|\mathcal L(T)|\le J}
 \left\{\sum_{v\in\mathcal L(T)}\lambda_v\phi_v^g:
   \bflambda\in[0,B]^L,\ c\in[0,1]^{L-1},\
   \tau\in(0,\infty)^{L-1}\right\},\qquad \mathcal F_{J,B}=\overline{\mathcal E_{J,B}}^{\,L^1(\mu)}.
\]
Here the bar denotes closure in $L^1(\mu)$. For the $k$th sieve,
we choose the bounds $J_k$ and $B_k$ as follows. Fix $\epsilon>0$
for C2--C3 and set
\[
\Delta=1-e^{-\epsilon^2/2},\qquad \delta=\Delta^2/128,\quad J_k=\max\left\{
   1,\left\lfloor\frac{\delta k}{\log(k+1)}\right\rfloor
 \right\},\qquad
 B_k=\frac{2\xi k}{\kappa},\qquad
 \mathcal F_k=\mathcal F_{J_k,B_k}.
\]
Since $J_k=\delta k/\log(k+1)+O(1)$, we have 
\begin{equation}\label{eqn: sievelimits}
 \frac{J_k\log k}{k}\to\delta,\qquad
 \frac{J_k\log J_k}{k}\to\delta,\qquad
 \frac{J_k}{k}\to0.
\end{equation}

\subsection{Verification of Condition C2}\label{app: C2}

\begin{lemma}[Tail bound for the number of leaves]\label{lem: tree tail}
Under Assumption~\ref{ass: tree prior}, $T$ is finite $\pi$-almost surely
and, for constants $C_T,c_T>0$,
\[
 \pi\{|\mathcal L(T)|>J\}\le C_Te^{-c_TJ\log J},
 \qquad J\in\mathbb N,\quad J\ge2.
\]
\end{lemma}

\begin{proof}
If $Z_h$ counts nodes at depth $h$, the independent splitting rule gives
\[
 \mathbb{E}_\pi[Z_h]
 =\prod_{i=0}^{h-1}2\alpha_0^{i+1}
 =2^h\alpha_0^{h(h+1)/2},\qquad
 \pi\{|T|=\infty\}
 \le\pi(Z_h>0)\le \mathbb{E}_\pi[Z_h]\to0.
\]
The rule $\rho(v)=\alpha_0^{1+\operatorname{depth}(v)}$,
$0<\alpha_0<1/2$, belongs to the depth-exponential family in
\citet[Section~5.2, (20)]{Rockova2019}. Their Corollary~5.2 gives the
tail bound. The independent axis and cut distributions do not alter the
branching law. The result follows by enlarging $C_T$ to cover all integers $J\ge2$.
\end{proof}

\begin{lemma}[Verification of C2]\label{lem: sieve}
The sets $\mathcal F_k$ are Borel in $\Theta$. Under
Assumptions~\ref{ass: tree prior} and \ref{ass: hyperparameters},
for every $0<c_1<\min\{c_T\delta,1\}$,
\[
 \Pi(\mathcal F_k^c)\le e^{-c_1k}
 \qquad\text{for all sufficiently large }k.
\]
\end{lemma}

\begin{proof}
The sets $\mathcal F_k$ are $L^1(\mu)$-closed, hence Borel.
Each model gate belongs to this class: take
$\tau_v=1/\gamma_{j_v}$ and its original cut $c_v$, so that
\[
 G_{c_v,\tau_v}(x_{j_v})
 =\{1+e^{-\gamma_{j_v}(x_{j_v}-c_v)}\}^{-1}.
\]
Thus
$\{|\mathcal L(T)|\le J_k,\
   \max_{v\in\mathcal L(T)}\lambda_v\le B_k\}$
implies $\lambda\in\mathcal F_k$.
For $Z\sim\mathrm{Ga}(\kappa,\kappa/\xi)$, Markov's inequality gives
\[
 \mathbb{E}[e^{\kappa Z/(2\xi)}]=2^\kappa,\qquad
 \Pr(Z>B_k)\le2^\kappa e^{-k}.
\]
The conditional union bound and Lemma~\ref{lem: tree tail} yield
\begin{equation}\label{eqn: sieveprior}
 \Pi(\mathcal F_k^c)
 \le\pi\{|\mathcal L(T)|>J_k\}+J_k\Pr(Z>B_k)
 \le C_Te^{-c_TJ_k\log J_k}+2^\kappa J_ke^{-k}
\end{equation}
for all large $k$. It follows from  \eqref{eqn: sievelimits} that
\[
 \limsup_{k\to\infty}\frac1k\log\Pi(\mathcal F_k^c)
 \le\max\{-c_T\delta,-1\}<-c_1.
\]
\end{proof}

\subsection{Verification of Condition C3}\label{app: C3}

Write $N(t,\mathcal F,d)$ for the minimum number of closed $d$-balls
of radius $t$, centered in $\mathcal F$, needed to cover $\mathcal F$,
with value $\infty$ if no finite cover exists
\citep[Appendix~C.1]{Ghosal2017}.

\begin{lemma}[Verification of C3]\label{lem: tests}
For the sieves $\mathcal F_k$, there exist measurable tests
$\psi_k\in[0,1]$ such that, for all sufficiently large $k$,
\[
 \mathbb{E}_0^{(k)}[\psi_k]\le e^{-k\Delta^2/16},\qquad
 \sup_{\lambda\in\mathcal F_k\cap A_\epsilon}
       \mathbb{E}_\lambda^{(k)}[1-\psi_k]
 \le e^{-k\Delta^2/8}.
\]
\end{lemma}

\begin{proof}
We cover the gates and leaf coefficients, bound the resulting sieve
entropy, and combine local tests. The product-cover strategy follows
\citet[Supplement, Section~S.2, proof of Lemma~1(1)]{Linero2018};
the testing argument follows \citet[proof of Theorem~6.23]{Ghosal2017}.

\medskip
\noindent\emph{Step 1. Cover the gates.}
Let $\mu_j$ be the $j$th coordinate marginal of the probability measure $\mu/V$.
Set
\[
 \mathcal G=\left\{G_{c,\tau} (y):=\frac{1}{1 + \exp\{ - (y-c)/\tau \}}:c\in[0,1],\ \tau>0 \right\},\qquad
 \mathcal H=\{y\mapsto by+c:(b,c)\in\mathbb R^2\}.
\]
By \citet[Lemma~2.6.15]{Vaart1996}, $\dim(\mathcal H)=2$ implies
subgraph VC index at most $4$. For $0<a<1$,
\[
 a<G_{c,\tau}(y)
 \quad\Longleftrightarrow\quad
 \log\frac a{1-a}<\frac{y-c}{\tau},\qquad
 \left(y\mapsto\frac{y-c}{\tau}\right)\in\mathcal H.
\]
Membership is fixed for $a\notin(0,1)$, so $\mathcal G$ also has
subgraph VC index at most $4$.
Applying Theorem~2.6.7 of \citet{Vaart1996}, with envelope $1$ and $r=1$,
yields, 
$$
N(u,\mathcal G,L^1(\mu_j)) \leq C_0 u^{-3}, \quad 0<u<1
$$
for a universal constant $C_0$.

To ensure the covering centers belong to $\mathcal{G}$, we can take $u= \eta/2$ for $\eta\in (0,1)$. Then  by covering property, there exists functions $f_1, \ldots, f_m$ such that $m \leq C_0 (\eta/2)^{-3}$ and 
$$
\mathcal G \subset \cup_{i=1}^m \{ f: \| f- f_i\|_{L^1(\mu_j)} \leq \eta/2\}.
$$

Discard balls that do not intersect $\mathcal G$ and relabel the remaining centers. Their number still satisfies the same bound. From each remaining ball choose $G_i\in \mathcal{G}$ such that $\|G_i - f_i\|_{L^1(\mu_j)} \leq \eta/2$. For any $G\in \mathcal{G}$, select $i$ satisfying $\|G-f_i\|_{L^1(\mu_j)} \leq \eta/2$ then 
$$
\| G-G_i\|_{L^1(\mu_j)} \leq \| G - f_i\|_{L^1(\mu_j)} + \| f_i - G_i\|_{L^1(\mu_j)} \leq \eta.
$$
Thus $G_1, \ldots, G_m \in \mathcal G$ form an $\eta$-cover of $\mathcal G$ and 
$$
N(\eta,\mathcal G,L^1(\mu_j)) \leq m \leq  C_{\rm VC}  \eta^{-3}. 
$$
for $\eta \in (0,1)$, where $C_{\rm VC}=\max\{1,8C_0\}$.

\medskip
\noindent\emph{Step 2. Cover the sieve.}
Fix $t>0$ and consider $\mathcal F_k=\mathcal F_{J_k,B_k}$. Since
\[
 0<\eta:=\frac{t^2}{2B_kVJ_k}
 =\frac{\kappa t^2}{4\xi VkJ_k}
 \le\frac{\kappa t^2}{4\xi Vk}\longrightarrow0,
\]
take $k$ sufficiently large that $\eta<1$ throughout this step.
Fix $(T,\boldsymbol j)$ with $L\le J_k$ leaves. Write
\[
 f=\sum_{v\in\mathcal L(T)}\lambda_v\phi_v^g,\qquad
 f'=\sum_{v\in\mathcal L(T)}\lambda_v'\phi_v^{g'}.
\]
If the leaf coefficients agree and $g,g'$ differ only at node $v$,
the routing recursion gives
\[
 f-f'
 =\phi_v(g_v-g_v')(f_{v,\mathrm r}-f_{v,\mathrm l}),\qquad
 0\le\phi_v\le1,\quad
 0\le f_{v,\mathrm l},f_{v,\mathrm r}\le B_k,
\]
where $f_{v,\mathrm l},f_{v,\mathrm r}$ are the child-subtree intensities,
with initial routing weight $1$ at each child.

Enumerate the internal nodes as $v_1,\ldots,v_m$, where $m=L-1$. Let $f^{(i)}$ use the new gates at $v_1,\ldots,v_i$, the original gates elsewhere, and the original leaf coefficients. Then 
$$
f^{(0)} = f, \, f^{(m)} = \sum_{v\in\calL} \lambda_v \phi_v^{g'},\, f' =  \sum_{v\in\calL} \lambda_v' \phi_v^{g'}.
$$
Observe that $ | f^{(i-1)}(\bfx)  - f^{(i)}(\bfx) | \leq B_k | g_{v_i} (\bfx) - g'_{v_i}(\bfx)|$. Then telescoping sum yields 
\begin{align*}
\| f- f'\|_1 &= \int_D | f(\bfx) - f'(\bfx) | \mu(\mathrm{d}\bfx) \\
&= \int_D \left| \sum_{i=1}^m [f^{(i-1)}(\bfx) - f^{(i)}(\bfx)] + \sum_{v\in \calL} (\lambda_v - \lambda_v') \phi_v^{g'}(\bfx) \right| \mu(\mathrm{d} \bfx) \\
&\leq \sum_{i=1}^m \int_D | f^{(i-1)}(\bfx) - f^{(i)}(\bfx) | \mu(\mathrm{d}\bfx) + \int_D \sum_{v\in \mathcal L} |\lambda_v - \lambda_v'| \phi_v^{g'}(\bfx) \mu(\mathrm{d} \bfx) \\
& \leq B_k \sum_{v \in \calI} \| g_v - g_v'\|_1 + V \| \bflambda - \bflambda'\|_{\infty}
\end{align*} 
where the last inequality uses $\sum_{v \in \calL} \phi_v^{g'}(\bfx) = 1$ and $V=\int_D 1 \mu(\mathrm{d}\bfx)$.

We allocate at most $t^2/2$ to each term in the preceding
$L^1(\mu)$ bound. For the coefficient term, this requires
\[
 V\|\bflambda-\bflambda'\|_\infty\le\frac{t^2}{2}
 \quad\Longleftrightarrow\quad
 \|\bflambda-\bflambda'\|_\infty\le\frac{t^2}{2V}.
\]
Set
\[
 a=\frac{t^2}{2V},\qquad
 \mathcal A_t=\{ia:i=0,\ldots,\lfloor B_k/a\rfloor\}\cup\{B_k\}.
\]
The gap between consecutive grid points is at most $a$. For every
$\bflambda\in[0,B_k]^L$, choosing
$\lambda_v'=a\lfloor\lambda_v/a\rfloor\in\mathcal A_t$ gives
\[
 0\le\lambda_v-\lambda_v'<a,\qquad
 V\|\bflambda-\bflambda'\|_\infty
 \le Va=V\frac{t^2}{2V}=\frac{t^2}{2}.
\]
The product grid for the $L$ coefficients has cardinality
\[
 |\mathcal A_t^L|=|\mathcal A_t|^L
 \le\left(\left\lfloor\frac{B_k}{a}\right\rfloor+2\right)^L
 \le\left(2+\frac{B_k}{a}\right)^L
 =\left(2+\frac{2B_kV}{t^2}\right)^L.
\]

For the gates, write $g_v(\bfx)=G_v(x_{j_v})$ and
$g_v'(\bfx)=G_v'(x_{j_v})$, with $G_v,G_v'\in\mathcal G$.
A gate cover of radius $\eta$
in $L^1(\mu_{j_v})$ permits a choice satisfying
\[
 \|g_v-g_v'\|_1
 =V\|G_v-G_v'\|_{L^1(\mu_{j_v})}\le V\eta.
\]
There are $L-1$ internal nodes, so the choice
$\eta=t^2/(2B_kVJ_k)$ ensures
\[
 B_k\sum_{v\in\mathcal I(T)}\|g_v-g_v'\|_1
 \le B_kV(L-1)\eta
 \le B_kVJ_k\eta
 =B_kVJ_k\frac{t^2}{2B_kVJ_k}=\frac{t^2}{2}.
\]
Since $0<\eta<1$, Step~1 bounds the number of choices at each node by
\begin{align*}
 N(\eta,\mathcal G,L^1(\mu_{j_v}))
 &\le C_{\rm VC}\eta^{-3}
 =C_{\rm VC}
   \left(\frac{1}{t^2/(2B_kVJ_k)}\right)^3
 =C_{\rm VC}\left(\frac{2B_kVJ_k}{t^2}\right)^3.
\end{align*}
Taking the product of these covers over the $L-1$ internal nodes
therefore gives at most
\[
 \left\{C_{\rm VC}\left(\frac{2B_kVJ_k}{t^2}\right)^3\right\}^{L-1}
\]
gate configurations. Every $f$ has an approximating intensity $f'$
formed from the coefficient grid and these gate covers, with
\[
 \|f-f'\|_1
 \le V\|\bflambda-\bflambda'\|_\infty
     +B_k\sum_{v\in\mathcal I(T)}\|g_v-g_v'\|_1
 \le\frac{t^2}{2}+\frac{t^2}{2}=t^2.
\]
The Catalan count \citep[Exercise~6.19(d)]{Stanley1999} gives
\[
 \sum_{L=1}^{J_k}\frac1L\binom{2L-2}{L-1}d^{L-1}
 \le J_k(4d)^{J_k-1}.
\]

Using $L\leq J_k$ in the coefficient and gate bound and multiplying by the topology and axis count yield 
$\lambda_1,\ldots,\lambda_N\in\mathcal E_{J_k,B_k}\subseteq\mathcal F_k$, with  $N$ bounded by the right-hand side of \eqref{eqn: entropy}. 

Note that $\forall f \in \mathcal E_{J_k,B_k}$, there exists $i \in \{1, \ldots, N\}: \| f-\lambda_i\|_1 \leq t^2$. Thus 
\[
 \mathcal E_{J_k,B_k}
 \subseteq\bigcup_{i=1}^N
 \{\lambda\in L^1(\mu):\|\lambda-\lambda_i\|_1\le t^2\}.
\]
For each $i$, the reverse triangle inequality gives
\[
 \bigl|\|\lambda-\lambda_i\|_1
       -\|\lambda'-\lambda_i\|_1\bigr|
 \le\|\lambda-\lambda'\|_1.
\]
Thus $\lambda\mapsto\|\lambda-\lambda_i\|_1$ is continuous in
$L^1(\mu)$, so each displayed ball is closed. Their union is closed
because $N<\infty$. Taking the $L^1(\mu)$ closure of the first inclusion
therefore yields
\[
 \mathcal F_k
 \subseteq\bigcup_{i=1}^N
 \{\lambda\in L^1(\mu):\|\lambda-\lambda_i\|_1\le t^2\}.
\]

For $\lambda,\lambda'\in\Theta$, we adapt the argument in the proof of
\citet[Lemma~B.1(i)--(ii), pp.~517--518]{Ghosal2017}, stated there for
probability densities, to the present nonnegative, integrable intensities.
First, for $a,b\ge0$,
\[
 |\sqrt a-\sqrt b|\le\sqrt a+\sqrt b,
 \qquad
 (\sqrt a-\sqrt b)^2
 \le|\sqrt a-\sqrt b|(\sqrt a+\sqrt b)=|a-b|.
\]
Taking $a=\lambda(\bfx)$ and $b=\lambda'(\bfx)$ for $\mu$-a.e. $\bfx$
and integrating gives
\begin{equation}\label{eqn: L1bound}
d_H^2(\lambda,\lambda')
 =\int_D(\sqrt\lambda-\sqrt{\lambda'})^2\,\mathrm d\mu \le\int_D|\lambda-\lambda'|\,\mathrm d\mu
 =\|\lambda-\lambda'\|_1.
\end{equation}
Moreover, integrability gives
\[
 \|\sqrt\lambda\|_2^2=\int_D\lambda\,\mathrm d\mu<\infty,
 \qquad
 \|\sqrt{\lambda'}\|_2^2=\int_D\lambda'\,\mathrm d\mu<\infty.
\]
Thus $|\sqrt\lambda-\sqrt{\lambda'}|$ and
$\sqrt\lambda+\sqrt{\lambda'}$ belong to $L^2(\mu)$.
It follows from Cauchy--Schwarz and the $L^2(\mu)$ triangle inequality that 
\begin{align*}
 \|\lambda-\lambda'\|_1
 &=\int_D|\sqrt\lambda-\sqrt{\lambda'}|
           (\sqrt\lambda+\sqrt{\lambda'})\,\mathrm d\mu\\
 &\le\left\{\int_D(\sqrt\lambda-\sqrt{\lambda'})^2\,\mathrm d\mu\right\}^{1/2}
      \left\{\int_D(\sqrt\lambda+\sqrt{\lambda'})^2\,\mathrm d\mu\right\}^{1/2}\\
 &=d_H(\lambda,\lambda')\|\sqrt\lambda+\sqrt{\lambda'}\|_2\\
 &\le d_H(\lambda,\lambda')
       \bigl(\|\sqrt\lambda\|_2+\|\sqrt{\lambda'}\|_2\bigr).
\end{align*}
In particular, \eqref{eqn: L1bound} gives
\[
 \|\lambda-\lambda_i\|_1\le t^2
 \quad\Longrightarrow\quad d_H(\lambda,\lambda_i)\le t.
\]
The previously constructed cover therefore satisfies
\[
 \mathcal F_k\subseteq\bigcup_{i=1}^N
 \{\lambda\in\Theta:d_H(\lambda,\lambda_i)\le t\},
 \qquad \lambda_i\in\mathcal F_k.
\]

For each topology and axis assignment with $L\le J_k$ leaves, the
coefficient and gate covers give at most
\[
 \left(2+\frac{2B_kV}{t^2}\right)^L
 (C_{\rm VC}\eta^{-3})^{L-1}
\]
centers. There are
$L^{-1}\binom{2L-2}{L-1}d^{L-1}$ such topology--axis combinations.
The resulting closed $L^1(\mu)$-balls also cover $\mathcal F_k$ and lead to a Hellinger cover of radius $t$ since $d_H^2\le\|\cdot\|_1$.

For each fixed $t>0$ and sufficiently large $k$, we have
$0<\eta<1$ and $C_{\rm VC}\eta^{-3}\ge1$. Thus
\begin{align}
 N(t,\mathcal F_k,d_H)
 &\le \sum_{L=1}^{J_k}
 \frac1L\binom{2L-2}{L-1}d^{L-1}
 \left(2+\frac{2B_kV}{t^2}\right)^L
 (C_{\rm VC}\eta^{-3})^{L-1} \notag \\
 &\le
 \left(2+\frac{2B_kV}{t^2}\right)^{J_k}
 (C_{\rm VC}\eta^{-3})^{J_k-1}
 \sum_{L=1}^{J_k}(4d)^{L-1} \notag \\ \label{eqn: entropy}
 &\le J_k(4d)^{J_k-1}
 \left(2+\frac{2B_kV}{t^2}\right)^{J_k}
 \left\{C_{\rm VC}
 \left(\frac{2B_kVJ_k}{t^2}\right)^3\right\}^{J_k-1},
\end{align}
where we used
$L^{-1}\binom{2L-2}{L-1}\le4^{L-1}$ and
$\eta=t^2/(2B_kVJ_k)$.

\medskip
\noindent\emph{Step 3. Bound the entropy growth.}
The leaf-intensity bound remains $B_k=2\xi k/\kappa$, as in Step~2.
For fixed $t>0$, this choice and $J_k\le k$ eventually give
\[
 \log(2+2B_kV/t^2)=\log k+O_t(1),\qquad
 \log(2B_kVJ_k/t^2)\le2\log k+O_t(1).
\]
Here the constants in $O_t(\cdot)$ may depend on $t$ and the fixed
model parameters, but not on $k$.

It follows from Eqs.~\eqref{eqn: sievelimits} and~\eqref{eqn: entropy} that
\[
\log N(t,\mathcal F_k,d_H)
 \le7J_k\log k+O_t(J_k+\log J_k),\qquad \limsup_{k\to\infty}\frac1k\log N(t,\mathcal F_k,d_H)
 \le7\delta=\frac{7\Delta^2}{128}<\frac{\Delta^2}{16}.
\]

\medskip
\noindent\emph{Step 4. Construct and combine tests.}
We use the finite-cover testing construction in
\citet[proof of Theorem~6.23, p.~142]{Ghosal2017}.
For probability densities $p,q$ with respect to $Q$, write
$d_{\rm TV}(p,q)=\tfrac12\|p-q\|_{L^1(Q)}$.
The Poisson affinity formula in \citet[Theorem~7]{Leskela2024},
together with \citet[Lemma~B.1(i)--(ii)]{Ghosal2017}, gives
\begin{equation}\label{eqn: probabilitybounds}
1-e^{-d_H^2(\lambda,\lambda')/2}
 \le d_{\rm TV}(p_\lambda,p_{\lambda'}),\qquad d_{\rm TV}(p_\lambda,p_{\lambda'})
 \le\sqrt{1-e^{-d_H^2(\lambda,\lambda')}}
 \le d_H(\lambda,\lambda').
\end{equation}

Recall $\Delta=1-e^{-\epsilon^2/2}$ and put
\[
 E_k=\mathcal F_k\cap A_\epsilon,
 \qquad N_k=N(\Delta/8,\mathcal F_k,d_H).
\]
By Step~3, $\log N_k\le k\Delta^2/16$ for all sufficiently large $k$.
If $E_k=\varnothing$, take $\psi_k=0$. Otherwise, take an $N_k$-ball
cover of $\mathcal F_k$ of radius $\Delta/8$. Retain only balls
intersecting $E_k$, let $M_k$ be their number, and choose a new center
$\lambda_{k,j}$ in each intersection. The triangle inequality increases
the covering radius to at most $\Delta/8+\Delta/8=\Delta/4$, so
\[
 E_k\subseteq\bigcup_{j=1}^{M_k}
 \{\lambda\in\Theta:d_H(\lambda,\lambda_{k,j})\le\Delta/4\},
 \qquad \lambda_{k,j}\in E_k,\quad M_k\le N_k.
\]
For each $j=1,\ldots,M_k$, \eqref{eqn: probabilitybounds} gives
\[
d_{\rm TV}(p_0,p_{\lambda_{k,j}})
 \ge1-e^{-d_H^2(\lambda_0,\lambda_{k,j})/2}>\Delta,\qquad d_H(\lambda,\lambda_{k,j})\le\Delta/4
 \quad\Longrightarrow\quad
 d_{\rm TV}(p_\lambda,p_{\lambda_{k,j}})\le\Delta/4<\Delta/2.
\]

On the convex class of all $Q$-densities, total variation balls
are convex, since $d_{\rm TV}(p,q)=\tfrac12\|p-q\|_{L^1(Q)}$.
Moreover, $d_{\rm TV}(p,q)\le\|\sqrt p-\sqrt q\|_{L^2(Q)}$
by \citet[Lemma~B.1(i)]{Ghosal2017}. Thus it follows from 
\citet[Proposition~D.8, p.~540]{Ghosal2017} that 
\[
 d=d_{\rm TV},\quad p=p_0,\quad q=p_{\lambda_{k,j}},\quad
 n=k,\quad\varepsilon=\Delta,\quad c=1,
\]
which yields measurable tests $\psi_{k,j}=\psi_{k,j}(X^{(k)})\in[0,1]$.

By the preceding implication, these tests satisfy
\[
 \mathbb{E}_0^{(k)}[\psi_{k,j}]\le e^{-k\Delta^2/8},\qquad
 \sup_{\substack{\lambda\in\Theta\\
                  d_H(\lambda,\lambda_{k,j})\le\Delta/4}}
 \mathbb{E}_\lambda^{(k)}[1-\psi_{k,j}]\le e^{-k\Delta^2/8}.
\]
Define $\psi_k=\max_{1\le j\le M_k}\psi_{k,j}$.
Using $\psi_k\le\sum_j\psi_{k,j}$, $1-\psi_k\le1-\psi_{k,j}$,
and the cover of $E_k$ yields 
\begin{align*}
 \mathbb{E}_0^{(k)}[\psi_k]
 &\le M_ke^{-k\Delta^2/8}
 \le\exp\{\log N_k-k\Delta^2/8\}
 \le e^{-k\Delta^2/16},\\
 \sup_{\lambda\in E_k}\mathbb{E}_\lambda^{(k)}[1-\psi_k]
 &\le\max_{1\le j\le M_k}
 \sup_{\substack{\lambda\in\Theta\\
                  d_H(\lambda,\lambda_{k,j})\le\Delta/4}}
 \mathbb{E}_\lambda^{(k)}[1-\psi_{k,j}]
 \le e^{-k\Delta^2/8}.
\end{align*}
Thus C3 holds with $c_2=\Delta^2/16$ and $c_3=\Delta^2/8$.
\end{proof}

\subsection{Proof of Theorem~\ref{thm: posterior consistency}} \label{app: consistency} 

\begin{proof}[Proof of Theorem~\ref{thm: posterior consistency}]

We first verify measurability of the KL neighborhoods. Suppose
$\lambda_n,\lambda\in\Theta$ and
$\|\lambda_n-\lambda\|_{L^1(\mu)}\to0$. It follows from Equations~\eqref{eqn: L1bound} and \eqref{eqn: probabilitybounds} that 
\[
 \|p_{\lambda_n}-p_\lambda\|_{L^1(Q)}
 \le2d_H(\lambda_n,\lambda)
 \le2\|\lambda_n-\lambda\|_{L^1(\mu)}^{1/2}\longrightarrow0.
\]

Note that \citet[Lemma~B.4, p.~520]{Ghosal2017} proves
\[
 K(p_0;p_\lambda)
 \le\liminf_{n\to\infty}K(p_0;p_{\lambda_n}).
\]

For $a\in\mathbb R$, set
$C_a=\{\lambda\in\Theta:K(p_0;p_\lambda)\le a\}$.
The preceding inequality  implies
\[
 \lambda_n\in C_a,\quad
 \|\lambda_n-\lambda\|_{L^1(\mu)}\to0
 \quad\Longrightarrow\quad
 K(p_0;p_\lambda)\le\liminf_nK(p_0;p_{\lambda_n})\le a.
\]
Thus $C_a$ is closed in the metric space
$\Theta$. The Borel $\sigma$-field $\mathcal B(\Theta)$ contains closed
sets and their countable unions; therefore, for every $\eta>0$,
\[
 \{\lambda\in\Theta:K(p_0;p_\lambda)<\eta\}
 =\bigcup_{m=1}^{\infty}C_{\eta-1/m}
 \in\mathcal B(\Theta).
\]

It follows from Lemma~\ref{lem: KL support} and 
Assumptions~\ref{ass: truth}--\ref{ass: hyperparameters}, that  C1 holds true: 
\[
 \Pi\{\lambda:K(p_0;p_\lambda)<\eta\}>0
 \qquad(\eta>0).
\]

Fix $\epsilon>0$ and use the sieve construction and
Lemmas~\ref{lem: sieve}--\ref{lem: tests} with radius $\epsilon/2$.
To specify the density model in the consistency theorem, let
\[
 \mathcal P=\{p_\lambda:\lambda\in\Theta\}\subset L^1(Q),
\]
with densities identified $Q$-almost everywhere and $\mathcal P$
equipped with the subspace topology inherited from $L^1(Q)$.
The map
\[
 (\Theta,\|\cdot\|_{L^1(\mu)})
 \longrightarrow(\mathcal P,\|\cdot\|_{L^1(Q)}),
 \qquad\lambda\longmapsto p_\lambda,
\]
 is one-to-one by
\eqref{eqn: probabilitybounds}. Therefore, 
\[
 p_\lambda=p_{\lambda'}\quad Q\text{-a.e.}
 \quad\Longrightarrow\quad
 1-e^{-d_H^2(\lambda,\lambda')/2}=0
 \quad\Longrightarrow\quad
 \lambda=\lambda'\quad\mu\text{-a.e.}
\]
Its continuity follows from the first display of the proof. 

Conversely,
it follows from \eqref{eqn: probabilitybounds} and the Cauchy--Schwarz bound proved in
Lemma~\ref{lem: tests} that 
\begin{gather*}
 \|p_{\lambda_n}-p_\lambda\|_{L^1(Q)}\to0
 \quad\Longrightarrow\quad
 0\le1-e^{-d_H^2(\lambda_n,\lambda)/2}
 \le\tfrac12\|p_{\lambda_n}-p_\lambda\|_{L^1(Q)}\to0,\\
 d_H(\lambda_n,\lambda)\to0,\qquad
 \|\lambda_n-\lambda\|_{L^1(\mu)}
 \le d_H(\lambda_n,\lambda)
 \{d_H(\lambda_n,\lambda)+2\|\sqrt\lambda\|_{L^2(\mu)}\}
 \to0.
\end{gather*}
Thus the inverse is continuous. The displayed map is a homeomorphism
between these two metric spaces. In particular, it carries Borel sets
of $\Theta$ to Borel sets of $\mathcal P$.

Define, within $\mathcal P$,
\[
U_\epsilon=\{p_\lambda:d_H(\lambda,\lambda_0)<\epsilon\},\qquad \mathcal P_{k,1}=\{p_\lambda:\lambda\in\mathcal F_k\},\qquad
 \mathcal P_{k,2}=\{p_\lambda:\lambda\in\mathcal F_k^c\}.
\]
The preceding bounds show that $d_H$ and the $L^1(\mu)$ metric
induce the same topology on $\Theta$. Consequently,
$(U_{1/m})_{m\ge1}$ is a countable open neighborhood base at $p_0$,
and $\mathcal P_{k,1},\mathcal P_{k,2}$ form a Borel partition of
$\mathcal P$. Moreover,
\begin{align*}
 \{\lambda:p_\lambda\in\mathcal P_{k,1}\cap U_\epsilon^c\}
 &=\mathcal F_k\cap\{\lambda:d_H(\lambda,\lambda_0)\ge\epsilon\} \subseteq\mathcal F_k\cap A_{\epsilon/2}.
\end{align*}
Choose $0<C<\min\{c_1,c_2,c_3\}$, with the constants from C2--C3
at radius $\epsilon/2$. For all sufficiently large $k$, we have 
\begin{align*}
 \mathbb{E}_0^{(k)}[\psi_k]
 &\le e^{-c_2k}<e^{-Ck},\\
 \sup_{p_\lambda\in\mathcal P_{k,1}\cap U_\epsilon^c}
 \mathbb{E}_\lambda^{(k)}[1-\psi_k]
 &\le\sup_{\lambda\in\mathcal F_k\cap A_{\epsilon/2}}
       \mathbb{E}_\lambda^{(k)}[1-\psi_k]
 \le e^{-c_3k}<e^{-Ck},\\
 \Pi\{\lambda:p_\lambda\in\mathcal P_{k,2}\}
 &=\Pi(\mathcal F_k^c)\le e^{-c_1k}<e^{-Ck}.
\end{align*}

Define the induced prior on $\mathcal P$ by
\[
 \Pi_{\mathcal P}(B)
 :=\Pi\{\lambda:p_\lambda\in B\},
 \qquad B\in\mathcal B(\mathcal P).
\]
By C1 and the preceding bounds,
\citet[Theorem~6.17, final assertion, pp.~139--140]{Ghosal2017}
applies. Bayes' formula then gives
\[
 \Pi\{\lambda:d_H(\lambda,\lambda_0)\ge\epsilon\mid X^{(k)}\}
 =\Pi_{\mathcal P}(U_\epsilon^c\mid X^{(k)})
 \longrightarrow0
 \qquad\mathbb P_0^\infty\text{-a.s.}
\]

For $m\ge1$, put
\[
 E_m=\{\lim_{k\to\infty}\Pi(A_{1/m}\mid X^{(k)})=0\}.
\]
Then
\[
 \mathbb{P}_0^\infty(E_m)=1,\qquad
 \mathbb{P}_0^\infty\!\left(\bigcap_{m=1}^{\infty}E_m\right)=1.
\]
On this intersection, for every $\epsilon>0$ choose $m$ such that
$1/m<\epsilon$. Since $A_\epsilon\subseteq A_{1/m}$,
\[
 0\le\Pi(A_\epsilon\mid X^{(k)})
 \le\Pi(A_{1/m}\mid X^{(k)})\longrightarrow0.
\]
This proves strong consistency as defined in
Definition~\ref{def: consistency}.
\end{proof}

\end{document}

%% file: commands.tex
\newcommand{\bfx}{\mathbf{x}}
\newcommand{\bfy}{\mathbf{y}}

\newcommand{\bflambda}{\boldsymbol \lambda}

\newcommand{\calF}{{\cal F}}

\newcommand{\calI}{{\cal I}}

\newcommand{\calL}{{\cal L}}

\newcommand{{\Pa}}{\text{Pa}}
\newcommand{{\Phild}}{\text{child}}
\newcommand{{\sigmoid}}{\text{sigmoid}}

\newcommand{{\lchild}}{\mathrm{l}}
\newcommand{{\rchild}}{\mathrm{r}}